\documentclass{cccg26}
\usepackage{graphicx,amssymb,amsmath}
\usepackage{enumitem}
\usepackage{tikz}
\usepackage{pgfplots}
\pgfplotsset{compat=1.18}
\usepackage[colorlinks=true,
            linkcolor=black,
            citecolor=black,
            urlcolor=blue]{hyperref}

\title{Minimum-Weight Steiner Triangulation of Convex Polygons\\ Requires Interior Steiner Points}
\author{
David Eppstein\thanks{Department of Computer Science, University of California, Irvine, \texttt{eppstein@uci.edu}}
\and
Zahra Hadizadeh\thanks{Department of Computer Science, University of California, Irvine, \texttt{zhadizad@uci.edu}}
}

\index{Hadizadeh, Zahra}
\index{Eppstein, David}

\begin{document}
\thispagestyle{empty}
\maketitle

\begin{abstract}
We construct a convex polygon for which the minimum-weight Steiner triangulation requires an interior Steiner point. This provides a counterexample to a 1994 conjecture of  Eppstein that minimum-weight Steiner triangulation of convex polygons needs only Steiner points on the boundary of the polygon.
\end{abstract}

\section{Introduction}

We study the problem of \emph{minimum-weight Steiner triangulation}: triangulating a geometric input, allowing the addition of \emph{Steiner points} as vertices of the triangulation, with the goal of minimizing the total edge length of the triangulation. For the version of the problem studied here, the input is a convex polygon. Approximation algorithms for this problem, and for the related problem of minimum-weight Steiner triangulation of point sets, were considered by Eppstein~\cite{Eppstein94}. Eppstein observed that there exist convex polygons for which optimal or approximately-optimal triangulations require Steiner points: a triangulation without Steiner points is $\Omega(\log n)$ times as expensive as a triangulation with Steiner points. Additionally, he found quadtree-based approximation algorithms for minimum-weight Steiner triangulation, both for point sets and for convex polygons, that achieve a constant (but large) approximation ratio. In the case of convex polygons, his approximation uses only Steiner points on the boundary of the polygon. He conjectured that, in fact, only boundary Steiner points are necessary for the optimal solution to the minimum-weight Steiner triangulation on convex polygons. Triangulations that only use vertices on the boundary have a tree-structured weak dual (the graph of triangles and their adjacencies) through which it might be possible to find the optimal such triangulation in polynomial time by dynamic programming.

In this paper, we give a counterexample to this conjecture. We construct a convex polygon for which adding a single interior Steiner point strictly improves over all triangulations that use only boundary Steiner points (and also over triangulations with no Steiner points). This shows that interior Steiner points can be strictly more powerful than boundary-only Steiner points, even for convex polygons.

In any counterexample to this conjecture, with one Steiner point (as ours has), the Steiner point must have high degree ($\ge 7$), because a Steiner point of degree at most six could be contracted to its closest neighbor, reducing the total weight of the triangulation by the triangle inequality. The intuitive idea behind our counterexample is to construct a convex polygon whose minimum-weight non-Steiner triangulation has a vertex $v_0$ that also has high degree, with neighbors at exponentially increasing distances. For this triangulation, an interior Steiner point near $v_0$ can take over many of the adjacencies to neighbors of $v_0$, shortening those edges, while adding only a small amount of length between the new vertex and $v_0$. The difficult part of our construction is proving that, for this example, one cannot obtain this reduction in length using Steiner points that belong only to the polygon boundary.

\subsection{New results}

Our main result is the following theorem.

\begin{theorem}\label{thm:main}
There exists a convex polygon $P$ such that:
\begin{enumerate}[label=(\roman*)]
    \item \label{thm:part1} the minimum cost of a triangulation of $P$ without Steiner points is
    \[
    \mathrm{OPT}_{\mathrm{noSteiner}}(P);
    \]
    \item \label{thm:part2} allowing Steiner points only on the boundary does not improve this value, i.e.,
    \[
    \mathrm{OPT}_{\mathrm{boundary}}(P)= \mathrm{OPT}_{\mathrm{noSteiner}}(P);
    \]
    \item \label{thm:part3} allowing interior Steiner points strictly improves the optimum, i.e.,
    \[
    \mathrm{OPT}_{\mathrm{interior}}(P)
    < \mathrm{OPT}_{\mathrm{boundary}}(P).
    \]
\end{enumerate}
\end{theorem}

The polygon $P$ is given explicitly in the next section. The interior improvement is achieved by adding a single Steiner point. Together, these statements give a counterexample to the conjecture of Eppstein~\cite{Eppstein94}.

\subsection{Related work}

Minimum-weight triangulation of point sets, without Steiner points, was studied by several groups of researchers in the 1970s~\cite{DupGot-AV-70,ShaHoe-FOCS-76,Llo-FOCS-77}, and included in a list of problems neither known to be polynomial time nor $\mathsf{NP}$-complete by Garey and Johnson in 1979~\cite{GarJoh-79}. This open problem was finally resolved by an $\mathsf{NP}$-completeness proof published by Mulzer and Rote in 2008~\cite{MulzerRote08}. For simple polygons, a standard dynamic programming algorithm finds the minimum-weight non-Steiner triangulation in polynomial time~\cite{Gil-RR-79,Kli-ADM-80}.

The use of Steiner points in triangulation is commonplace in finite element mesh generation, but often with different optimization criteria than minimum weight. The problem of minimum-weight Steiner triangulation of polygons was introduced by Clarkson, who provided a logarithmic approximation~\cite{Cla-SODA-91}.
Cynthia Traub has studied the local topological effects of Steiner points in minimum-weight triangulations~\cite{Traub15}.
The constant-factor approximation algorithm for minimum-weight Steiner triangulation of point sets and convex polygons of Eppstein~\cite{Eppstein94} has been extended more generally to polygons with holes~\cite{ChengLee02}, and to minimum-area Steiner triangulation of three-dimensional inputs~\cite{ChengDey99}.

\section{Proof of Theorem~\ref{thm:main}}

We prove the theorem for one explicit convex polygon \(P\).

Let
\[
v_0=(0,0), \qquad v_{14}=(16000,0),
\]
and for \(1\le i\le 13\) define
\[
v_i=\left(\frac{2^i-i}{5},\,2^i+1\right),
\qquad
v_{28-i}=\left(\frac{2^i-i}{5},\,-(2^i+1)\right).
\]
These are the vertices of \(P\) in counterclockwise order.

\paragraph{Geometric intuition.}
The construction has two deliberately engineered features.

First, on the left side of the polygon, the vertices grow
approximately exponentially in their \(y\)-coordinate while their
\(x\)-coordinate grows only slowly. This creates a very steep convex
chain. As a consequence, diagonals between nonconsecutive vertices are
expensive. The cheapest triangulation without Steiner points is
therefore the fan centered at the origin.

Second, the rightmost vertex
\[
e=(16000,0)
\]
is placed very far from the rest of the polygon. This makes the two
incident hull edges unusually long and shallow. Because of this poor
geometry, adding one carefully chosen interior Steiner point near the
origin slightly shortens \emph{all} connections simultaneously, while
adding Steiner points only on the boundary cannot exploit this global
improvement.

Figure~\ref{fig:construction-intuition} illustrates the point set and
the interior Steiner point used in the improvement construction.

% preamble:
% \usepackage{tikz}
% \usepackage{pgfplots}
% \pgfplotsset{compat=1.18}

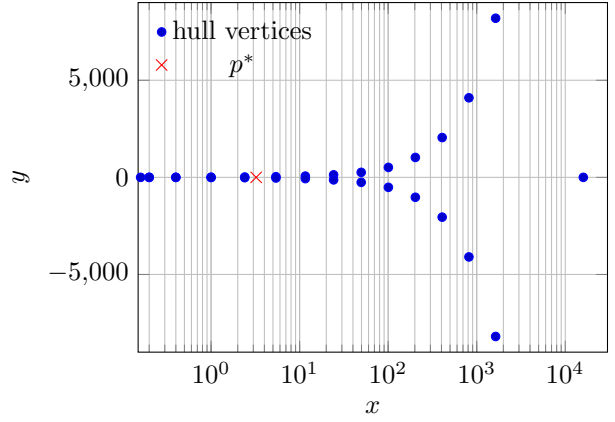
\begin{figure}[t]
\centering
\begin{tikzpicture}
\begin{axis}[
    width=0.92\linewidth,
    height=6.2cm,
    xlabel={$x$},
    ylabel={$y$},
    grid=both,
    xmode=log,
    log basis x=10,
    xmin=0.15,
    xmax=30000,
    ymin=-9000,
    ymax=9000,
    legend style={at={(0.02,0.98)},anchor=north west,draw=none,fill=none}
]

\addplot+[only marks, mark=*, mark size=1.6pt] coordinates {
(0.16,0)
(0.2,3) (0.4,5) (1,9) (2.4,17) (5.4,33)
(11.6,65) (24.2,129) (49.6,257) (100.6,513)
(202.8,1025) (407.4,2049) (816.8,4097) (1635.8,8193)
(0.2,-3) (0.4,-5) (1,-9) (2.4,-17) (5.4,-33)
(11.6,-65) (24.2,-129) (49.6,-257) (100.6,-513)
(202.8,-1025) (407.4,-2049) (816.8,-4097) (1635.8,-8193)
(16000,0)
};
\addlegendentry{hull vertices}

\addplot+[only marks, mark=x, mark size=3pt] coordinates {
(3.23041167,0)
};
\addlegendentry{$p^*$}

\end{axis}
\end{tikzpicture}
\caption{Point set shown with log scale on the \(x\)-axis. The point \((0,0)\) is displayed slightly shifted only for visualization.}
\label{fig:construction-intuition}
\end{figure}

We now prove the three parts separately.

\subsection{Proof of Theorem~\ref{thm:main}, part~\ref{thm:part1}}
\label{sec:proof-part1}

We compute the optimal triangulation of \(P\) without Steiner points
using the standard dynamic programming formulation for convex polygons.

Define
\[
c(i,j)=
\begin{cases}
0,&\text{if }(v_i,v_j)\text{ is a hull edge},\\
\|v_i-v_j\|_2,&\text{otherwise.}
\end{cases}
\]
Then
\[
\mathrm{DP}[i][i+1]=0,
\]
and for \(j\ge i+2\),
\[
\mathrm{DP}[i][j]
=
\min_{i<k<j}
\Bigl(
\mathrm{DP}[i][k]
+
\mathrm{DP}[k][j]
+
c(i,k)
+
c(k,j)
\Bigr).
\]

Since \(P\) is convex, this recurrence computes the exact optimum.
Evaluating it on our \(28\)-vertex polygon yields
\[
\mathrm{OPT}_{\mathrm{noSteiner}}(P)=49425.560539790305.
\]
Moreover, the computed optimum is the fan triangulation centered at
\[
v_0=(0,0).
\]

Additional numerical details, including representative entries of the
dynamic programming table, are given in Appendix~\ref{app:dp}.
\subsection{Proof of Theorem~\ref{thm:main}, part~\ref{thm:part3} }

We now exhibit one interior Steiner point whose star triangulation has strictly smaller cost.

Let
\[
p^*=(3.23041167,0).
\]
Connect \(p^*\) to all vertices of \(P\). Since \(p^*\) lies strictly inside the polygon, this gives a valid triangulation. The total added length is
\[
\sum_{t=0}^{27}\|p^*-v_t\|_2 = 49425.421179956924.
\]
Therefore
\[
\mathrm{OPT}_{\mathrm{interior}}(P)
<
49425.56
<
\mathrm{OPT}_{\mathrm{noSteiner}}(P).
\]
This proves part~(iii).

\subsection{Proof of Theorem~\ref{thm:main}, part~\ref{thm:part2}}

In this subsection we prove that allowing Steiner points only on the boundary does not improve the optimum for the polygon \(P\).

We call the point \((16000,0)\) the \emph{extreme point}. The two hull edges incident to it are the \emph{extreme edges}. The line through \((8179/5,8193)\) and \((16000,0)\) is called the \emph{upper extreme line}, and the symmetric lower one is called the \emph{lower extreme line}.

\subsubsection*{Distance bounds}

We first record two simple bounds that will be used throughout.

\begin{obs}\label{obs:y-difference}
For any two points \(u,v\) on non-extreme hull edges of \(P\),
\[
d(u,v)\le 1.03\,|u_y-v_y|.
\]
\end{obs}
The proof is deferred to Appendix~\ref{proof:obs:y-difference}.

\begin{obs}\label{obs:dist-to-extreme-line}
Let \(L\) be the upper extreme line through \((8179/5,8193)\) and \((16000,0)\). Then:
\begin{enumerate}[label=(\roman*)]
    \item for every point \((x,y)\) on a non-extreme hull edge with
    \[
    0<y\le 4097,
    \]
    we have
    \[
    d((x,y),L)\ge 7924.6-0.968\,y;
    \]
    \item for every point \((x,y)\) on a non-extreme hull edge with \(y<0\), we have
    \[
    d((x,y),L)\ge 7924.6+0.769\,|y|.
    \]
\end{enumerate}
\end{obs}
The proof is deferred to Appendix~\ref{proof:obs:dist-to-extreme-line}.

\begin{obs}\label{obs:extreme-distance}
Let \(e=(16000,0)\).

\begin{enumerate}
    \item For every point \(x\) on a non-extreme hull edge,
    \[
    15500<d(e,x)<17000.
    \]

    \item If in addition
    \[
    |x_y|\le 4097,
    \]
    then
    \[
    d(e,x)\le 16000.
    \]
\end{enumerate}
\end{obs}
The proof is deferred to Appendix~\ref{proof:extreme-distance}.

\subsubsection*{No Steiner point on the short edges adjacent to the extreme edges}

\begin{lemma}\label{lem:no-short-edge}
In an optimal boundary-only Steiner triangulation, there is no Steiner point on the edge
\[
(4084/5,4097)\text{--}(8179/5,8193),
\]
and none on its symmetric lower edge
\[
(4084/5,-4097)\text{--}(8179/5,-8193).
\]
\end{lemma}
\begin{proof}
We prove the upper statement; the lower one follows by symmetry.

Let
\[
b=(8179/5,8193),
\]
and assume for contradiction that a Steiner point \(p\) lies on the segment
\[
(4084/5,4097)\text{--}b.
\]
Let \(a\) be the point immediately below \(p\) on this segment.

We first simplify the structure of the neighborhood of \(p\).

\medskip
\noindent
\emph{Reduction.}
If \(p\) has two neighbors on the upper extreme edge (besides \(b\)), then one of them is redundant: all of its incident edges can be removed without harming the triangulation. Hence, we may assume that \(p\) has at most one neighbor \(q\) on the upper extreme edge.

\medskip
We now reroute edges incident to \(p\).

\medskip
\noindent

Let \(x\) be a neighbor of \(p\) distinct from \(a\) and not on the upper extreme edge. Then
\[
d(a,x)\le d(p,x),
\]
so reconnecting \(x\) to \(a\) does not increase the cost (see Appendix~\ref{app:replace}).

\medskip
\noindent
\emph{Case 1: all remaining neighbors lie on extreme edges.}

Let
\[
S:=N(p)\setminus\{a,b,q\},
\]
and assume all points of \(S\) lie on the lower extreme edge.

If \(q\) is not the extreme vertex, we remove both \(p\) and \(q\), reconnect all vertices of \(S\) to \(a\), and add one diagonal to restore a triangulation. This strictly decreases the total cost (see Appendix~\ref{app:extreme-case}).

If \(q\) is the extreme vertex, we instead remove \(p\) only and reconnect all vertices of \(S\) to \(a\); this again yields a strictly cheaper triangulation (see Appendix~\ref{app:extreme-case}).

\medskip
\noindent
\emph{Case 2: there exists a neighbor not on an extreme edge.}

Let \(S_1\subseteq S\) be the set of such neighbors, and assume \(|S_1|\ge 1\). Set
\[
\Delta:=p_y-a_y>0.
\]

Reconnecting all vertices of \(S\) to \(a\) yields a total gain of at least \(d(a,p)\), while the only possible loss (replacing \(pq\) by \(aq\)) is at most \(d(a,p)\). Hence the total cost strictly decreases (see Appendix~\ref{app:mixed-case}).

\medskip
In all cases we obtain a strictly cheaper triangulation, contradicting optimality.
\end{proof}

\subsubsection*{No Steiner point on an extreme edge with neighbors only on one side}

\begin{lemma}\label{lem:no-one-sided-extreme}
In an optimal boundary-only Steiner triangulation, there is no Steiner point on an upper extreme edge all of whose neighbors lie on the upper hull.
\end{lemma}
\begin{proof}
Assume for contradiction that such a point \(p\) exists, and among all such points choose \(p\) to be the highest one on the extreme edges. Let \(S\) be the set of neighbors of \(p\), and let \(w\in S\setminus\{u,v\}\) be the neighbor of smallest \(y\)-coordinate.

We remove \(p\) and reroute all its incident edges through \(w\).
Intuitively, this operation shortens all connections: every edge from \(p\) to a vertex on the upper hull is long because \(p\) lies near the extreme edge, while rerouting through the lower point \(w\) produces significantly shorter connections.

The key point is that the total length of edges incident to \(p\) is large (each such edge has length bounded below using Observation~\ref{obs:dist-to-extreme-line}), while the cost of reconnecting through \(w\) is controlled by vertical distances and is therefore much smaller (using Observation~\ref{obs:y-difference}). 

A case analysis based on the position of \(w\) shows that the rerouted triangulation is strictly cheaper in all possible configurations. The full inequalities are given in Appendix~\ref{app:one-sided-extreme}.

This contradicts optimality.
\end{proof}

As an immediate consequence, there can be at most one Steiner point on the two extreme edges in total.

\subsubsection*{No Steiner point on an extreme edge at all}

\begin{lemma}\label{lem:no-extreme-edge}
In an optimal boundary-only Steiner triangulation, there is no Steiner point on either extreme edge.
\end{lemma}

\begin{proof}
Without loss of generality, assume for contradiction that a Steiner point \(p\) lies on the upper extreme edge. By Lemma~\ref{lem:no-one-sided-extreme}, the point \(p\) must have neighbors on both sides of the \(x\)-axis.

Let \(S\) be the set of neighbors of \(p\) excluding the extreme point \((16000,0)\), and partition
\[
\begin{aligned}
S=&S_1\dot\cup S_2,
\qquad\\
S_1=\{s\in S:s_y\geq 0\},
&\qquad
S_2=\{s\in S:s_y<0\}.
\end{aligned}
\]
Let \(u\in S\) be the highest point, let \(w_1\in S_1\) be the lowest nonnegative point, and let \(w_2\in S_2\) be the highest negative point.

\medskip
\noindent
\emph{Old cost.}
By Observation~\ref{obs:dist-to-extreme-line},
\[
\begin{aligned}
\sum_{x\in S} d(x,p)\ge&
\sum_{(x,y)\in S_1\setminus\{u\}} (7924.6-0.968y)\\
&+
\sum_{(x,y)\in S_2} (7924.6+0.769|y|).
\end{aligned}
\]

\medskip
\noindent
\emph{Rerouting through \(w_1\) or \(w_2\).}
For \(i\in\{1,2\}\), let \(N[w_i]\) denote \(w_i\) together with its two neighbors in \(S\). If we delete \(p\) and reroute through \(w_i\), the new cost is at most
\[
16000+\sum_{x\in S\setminus N[w_i]} d(x,w_i).
\]
Using Observation~\ref{obs:y-difference}, we obtain the bounds
\begin{align}
\mathrm{New}(w_1)
&\le
C_1
+1.03\!\!\sum_{x\in S\setminus N[w_1]} |x_y|
\notag\\
&\qquad
-1.03\,|w_{1,y}|\,(|S_1|-1-|S_2|), \label{eq:new-w1-main}\\
\mathrm{New}(w_2)
&\le
C_2
+1.03\!\!\sum_{x\in S\setminus N[w_2]} |x_y|
\notag\\
&\qquad
-1.03\,|w_{2,y}|\,(|S_2|-1-|S_1|), \label{eq:new-w2-main}
\end{align}
where
\[
C_1=
\begin{cases}
0, & \text{if } |S_1|=1,\\
16000, & \text{otherwise,}
\end{cases}
\qquad
C_2=
\begin{cases}
0, & \text{if } |S_2|=1,\\
16000, & \text{otherwise.}
\end{cases}
\]

The term \(C_i\) accounts for the possible new edge from \(w_i\) to the
extreme vertex \(e=(16000,0)\). If this edge is already present in the
triangulation, then \(C_i=0\). Otherwise \(w_i\) is not the topmost (or
bottommost) hull vertex, so
\[
|w_{i,y}|\le 4097.
\]
Hence by Observation~\ref{obs:extreme-distance},
\[
d(e,w_i)\le 16000,
\]
and charging \(16000\) is a valid upper bound.

\medskip
We now apply a sequence of reductions:
\begin{itemize}
    \item at most two points of \(S\) lie on each original hull edge;
    \item once \(w_1,w_2\) and \(|S_1|,|S_2|\) are fixed, all other points are chosen as high as possible;
    \item no additional point of \(S\) with \(|y|\le 2049\) is useful;
    \item by symmetry, we may assume \(|S_1|\ge |S_2|\);
    \item the possible locations of \(w_1,w_2\) are restricted according to the sign of \(|S_1|-|S_2|-1\).
\end{itemize}
These reductions are proved in Appendix~\ref{app:no-extreme-edge-reductions}.

After these reductions, only finitely many extremal or threshold configurations remain. The complete case analysis is given in Appendix~\ref{app:extreme-edge-cases}. In every case one verifies that
\[
\min\{\mathrm{New}(w_1),\mathrm{New}(w_2)\}
<
\sum_{x\in S}d(x,p).
\]
This contradiction shows that no Steiner point can lie on the upper extreme edge. The lower extreme edge is handled symmetrically.
\end{proof}

\subsubsection*{No Steiner point on any other hull edge}

We now turn to non-extreme hull edges.

\begin{lemma}\label{lem:one-point-per-edge}
On any non-extreme hull edge, there can be at most one boundary Steiner point.
\end{lemma}

\begin{proof}
Assume for contradiction that two Steiner points \(p,q\) lie on the same non-extreme hull edge \((u,v)\), with \(p\) between \(u\) and \(q\).

If one of \(p,q\) has all non-edge neighbors on one side of the hull edge and is not adjacent to the extreme vertex \(e=(16000,0)\), then that point can be deleted and its incident edges can be rerouted to the corresponding endpoint of \((u,v)\), strictly decreasing the cost.

Thus, by planarity, the only remaining case is that both \(p\) and \(q\) are adjacent to \(e\). In that case we delete both \(p,q\), reconnect the remaining neighbors to \(u\) and \(v\), and triangulate the remaining region. All rerouted edges become shorter, so it remains only to compare the two old edges \((e,p),(e,q)\) with the new diagonals.

The needed comparison is proved in Appendix~\ref{app:two-point-edge}:
\[
d(e,p)+d(e,q)
>
d(w_1,w_2)+\min\{d(w_1,v),d(w_2,u)\}.
\]
Hence the new triangulation is strictly cheaper, contradicting optimality.
\end{proof}

\begin{lemma}\label{lem:upper-reroute-gain}
Let \(u,v\) be a non-extreme hull edge, with \(u\) above \(v\), and let \(p\in uv\).
For any point \(x\) on the left side of the hull with \(x_y>u_y\),
\[
d(x,p)-d(x,u)\ge \frac{12}{13}\,d(u,p).
\]
\end{lemma}

\begin{proof}
The claim follows from a geometric bound on the angle \(\angle xup\).
Roughly speaking, vectors \(\overrightarrow{ux}\) and \(\overrightarrow{up}\)
form an obtuse angle bounded away from \(90^\circ\), which implies that
moving from \(p\) to \(u\) significantly shortens the distance to \(x\).

A complete proof, including the angle bound and the use of the law of
cosines, is given in Appendix~\ref{app:reroute-gain}.
\end{proof}

\begin{lemma}\label{lem:balanced-neighbors-no-extreme}
Assume a non-extreme hull edge \(uv\) contains exactly one Steiner point \(p\), and \(p\) is not adjacent to \(e\). Let \(S_1\) and \(S_2\) be the sets of neighbors above and below. Then
\[
|S_1|=|S_2|.
\]
\end{lemma}

\begin{proof}
Suppose \(|S_1|>|S_2|\). Delete \(p\) and reconnect all its neighbors to \(u\).

Each vertex in \(S_1\) gives a gain of more than
\[
\frac{12}{13}d(u,p),
\]
by Lemma~\ref{lem:upper-reroute-gain}, while each vertex in \(S_2\) causes a loss of at most \(d(u,p)\). By Claim~\ref{clm:min-side-size}, together with \(|S_1|>|S_2|\), we have \(|S_2|\le 12\). Hence the total gain is strictly larger than the total loss. This gives a cheaper triangulation, a contradiction. The case \(|S_2|>|S_1|\) is symmetric.
\end{proof}

\begin{lemma}\label{lem:degree-lower-bound}
Every boundary Steiner point on a non-extreme hull edge has at least three neighbors besides its endpoints.
\end{lemma}

\begin{proof}
Suppose \(p\) has only two neighbors \(w_1,w_2\).

One of them lies below \(v\), so \(d(p,w_2)\ge d(p,v)\). Also
\[
d(w_1,p)+d(p,v)>d(w_1,v).
\]
Thus replacing edges \(pw_1,pw_2\) by \(vw_1\) strictly reduces cost, a contradiction.
\end{proof}

\begin{lemma}\label{lem:no-consecutive}
There are no Steiner points on two consecutive non-extreme hull edges.
\end{lemma}
\begin{proof}
If neither point is adjacent to \(e\), planarity blocks neighbors on one side, contradicting Lemma~\ref{lem:degree-lower-bound}.

If both are adjacent to \(e\), remove both and retriangulate the resulting hexagon. All rerouted edges shorten, and the remaining comparison reduces to the extreme edges, which is handled in Appendix~\ref{app:consecutive-case}.
\end{proof}

\begin{lemma}\label{lem:no-nonextreme}
There is no boundary Steiner point on any non-extreme hull edge.
\end{lemma}

\begin{proof}
Assume for contradiction that a non-extreme hull edge \(uv\) contains a Steiner point \(p\), where \(u\) is the upper endpoint and \(v\) is the lower endpoint.

Let
\[
w_1,\dots,w_k
\]
be the neighbors of \(p\) above \(u\), ordered clockwise so that \(w_1\) is closest to \(u\), and let
\[
z_1,\dots,z_\ell
\]
be the neighbors below \(v\), ordered so that \(z_1\) is closest to \(v\). Let
\[
I_e=
\begin{cases}
1,&\text{if }e=(16000,0)\in N(p),\\
0,&\text{otherwise.}
\end{cases}
\]

If \(I_e=0\), then by Lemma~\ref{lem:balanced-neighbors-no-extreme},
\[
k=\ell.
\]
If \(I_e=1\), then both \(k,\ell\ge 1\).

\medskip
\noindent
\textbf{Rerouting to \(u\).}
We delete \(p\), reconnect all neighbors to \(u\), and if needed replace \(ep\) by \(eu\).

Upper neighbors give gain:
\[
\sum_{i=1}^k \bigl(d(w_i,p)-d(w_i,u)\bigr)
\ge \frac{12}{13}k\,d(p,u).
\]
Lower neighbors give loss:
\[
\sum_{j=1}^\ell \bigl(d(z_j,u)-d(z_j,p)\bigr)
\le \ell\,d(p,u).
\]
Extreme point contributes:
\[
d(e,u)-d(e,p)\le d(p,u),
\]
and if \(|p_y|,|u_y|\le 513\), then by Claim~\ref{clm:extreme-loss-small},
\[
d(e,u)-d(e,p)\le \tfrac{5}{13} d(p,u).
\]
Finally, since \(uw_1\) already exists,
\[
\text{extra gain} \;\ge d(u,w_1).
\]

Hence
\[
G_u
\ge
d(p,u)\Bigl(\tfrac{12}{13}k-\ell-I_e\Bigr)
+
d(u,w_1),
\]
and in the small-\(y\) case:
\[
G_u
\ge
d(p,u)\Bigl(\tfrac{12}{13}k-\ell-\tfrac{5}{13} I_e\Bigr)
+
d(u,w_1).
\]
Similarly, rerouting to \(v\) gives
\[
G_v
\ge
d(p,v)\Bigl(\tfrac{12}{13}\ell-k-I_e\Bigr)
+
d(v,z_1),
\]
and in the small-\(y\) case,
\[
G_v
\ge
d(p,v)\Bigl(\tfrac{12}{13}\ell-k-\tfrac{5}{13} I_e\Bigr)
+
d(v,z_1).
\]

\medskip
\noindent
\textbf{Case 1: \(k=\ell\).}
If \(uv\) is not incident to \(v_0\), reroute toward the endpoint with
larger absolute \(y\)-coordinate. By symmetry, assume this endpoint is
\(u\). Then
\[
d(u,w_1)\ge 2d(u,v).
\]
Thus
\[
G_u
\ge
d(p,u)\Bigl(-\tfrac{k}{13}-I_e\Bigr)
+
2d(u,v).
\]
Since \(k\le 13\), \(I_e\le 1\), and \(d(p,u)<d(u,v)\), we get
\[
G_u>0.
\]

Now suppose \(uv\) is incident to \(v_0\). By symmetry, assume
\[
v=v_0=(0,0), \qquad u=(1/5,3).
\]
Then \(|p_y|,|u_y|,|v_y|\le 513\), so the small-\(y\) bounds above apply
to both reroutings. Set
\[
C=\frac{k}{13}+\frac{5}{13}I_e.
\]
Since \(k\le 13\) and \(I_e\le 1\), we have \(C\le 18/13\). Because
\(k=\ell\), the two small-\(y\) bounds give
\[
G_u\ge -C\,d(p,u)+d(u,w_1),
\qquad
G_v\ge -C\,d(p,v)+d(v,z_1).
\]
If both were nonpositive, then
\[
d(u,w_1)+d(v,z_1)
\le
C(d(p,u)+d(p,v))
=
C\,d(u,v).
\]
But
\[
d(u,w_1)=\sqrt{2^2+(1/5)^2},
\qquad
d(v,z_1)=d(u,v)=\sqrt{3^2+(1/5)^2},
\]
so
\[
d(u,w_1)+d(v,z_1)
>
\frac{18}{13}d(u,v)
\ge
C\,d(u,v),
\]
a contradiction. Thus at least one of \(G_u,G_v\) is positive.

\medskip
\noindent
\textbf{Case 2: \(k>\ell\).}
Then
\[
G_u
\ge
d(p,u)\Bigl(\tfrac{12}{13}k-\ell-I_e\Bigr)
+
d(u,w_1).
\]
Using
\[
d(u,w_1)\ge \tfrac{12}{25}d(u,v)>\tfrac{12}{25}d(p,u),
\]
we obtain
\[
G_u
>
d(p,u)\Bigl(\tfrac{12}{13}k-\ell-I_e+\tfrac{12}{25}\Bigr).
\]
If \(|p_y|\le 513\), then
\[
G_u
>
d(p,u)\Bigl(\tfrac{12}{13}k-\ell-\frac{5}{13}I_e+\tfrac{12}{25}\Bigr).
\]

\medskip
\noindent
\textbf{Subcase 2a: \(|p_y|>513\).}
In this case, by Claim~\ref{clm:large-height-small-side}, we have
\[
\ell\le 5.
\]
Then
\[
\begin{aligned}
G_u
&>
d(p,u)\Bigl(\tfrac{12}{13}(\ell+1)-\ell-I_e+\tfrac{12}{25}\Bigr)\\
&=
d(p,u)\Bigl(\tfrac{12-\ell}{13}-I_e+\tfrac{12}{25}\Bigr)\\
&\ge
d(p,u)\Bigl(\tfrac{7}{13}-1+\tfrac{12}{25}\Bigr)\\
&=
d(p,u)\frac{6}{325}
>0.
\end{aligned}
\]

\medskip
\noindent
\textbf{Subcase 2b: \(|p_y|\le 513\).}
In this case, the gain satisfies
\[
\begin{aligned}
G_u
&>
d(p,u)\Bigl(\tfrac{12}{13}k-\ell-\frac{5}{13}I_e+\tfrac{12}{25}\Bigr)\\
&\ge
d(p,u)\Bigl(\tfrac{12}{13}(\ell+1)-\ell-\tfrac{5}{13}+\tfrac{12}{25}\Bigr)\\
&=
d(p,u)\Bigl(\tfrac{12-\ell}{13}-\tfrac{5}{13}+\tfrac{12}{25}\Bigr)\\
&\ge
d(p,u)\Bigl(-\tfrac{1}{13}-\tfrac{5}{13}+\tfrac{12}{25}\Bigr)>0 .
\end{aligned}
\]

\medskip
In all cases, one of the two reroutings gives positive gain, contradicting optimality.
\end{proof}
Therefore no optimal boundary-only triangulation contains a boundary Steiner point. Hence
\[
\mathrm{OPT}_{\mathrm{boundary}}(P)
=
\mathrm{OPT}_{\mathrm{noSteiner}}(P),
\]
proving part~(ii).

\section{Conclusions}
We have shown that the minimum-weight Steiner triangulation of a convex polygon may sometimes require interior Steiner points. This discovery shuts off the possibility of a dynamic programming algorithm for this problem based on the tree-like structure of triangulations without interior Steiner points. However, the computational complexity of minimum-weight Steiner triangulation of convex polygons remains unclear. Can it be shown to be $\mathsf{NP}$-hard, or polynomial time?

\section*{Acknowledgements}
This research was supported in part by NSF grant CCF-2212129.

\bibliographystyle{alpha} % or plain, abbrv, etc. 

\bibliography{Ref}

@article{Eppstein94,
  author       = {David Eppstein},
  title        = {Approximating the Minimum Weight {S}teiner Triangulation},
  journal      = {Discret. Comput. Geom.},
  volume       = {11},
  pages        = {163--191},
  year         = {1994},
  url          = {https://doi.org/10.1007/BF02574002},
  doi          = {10.1007/BF02574002},
  bibsource    = {dblp computer science bibliography, https://dblp.org}
}

@article {MulzerRote08,
    AUTHOR = {Mulzer, Wolfgang and Rote, G\"unter},
     TITLE = {Minimum-weight triangulation is {NP}-hard},
   JOURNAL = {J. ACM},
    VOLUME = {55},
      YEAR = {2008},
    NUMBER = {2},
     PAGES = {A11:1--A11:29},
  MRNUMBER = {2417038},
       DOI = {10.1145/1346330.1346336},
}

@article {ChengLee02,
    AUTHOR = {Cheng, Siu-Wing and Lee, Kam-Hing},
     TITLE = {Quadtree, ray shooting and approximate minimum weight
              {S}teiner triangulation},
   JOURNAL = {Comput. Geom.},
    VOLUME = {23},
      YEAR = {2002},
    NUMBER = {2},
     PAGES = {99--116},
  MRNUMBER = {1922926},
       DOI = {10.1016/S0925-7721(02)00078-0}
}

@inproceedings{ChengDey99,
  author       = {Siu{-}Wing Cheng and
                  Tamal K. Dey},
  title        = {Approximate minimum weight {S}teiner triangulation in three dimensions},
  booktitle    = {Proceedings of the Tenth Annual {ACM--SIAM} Symposium on Discrete Algorithms (SODA)},
  pages        = {205--214},
  publisher    = {ACM and SIAM},
  year         = {1999},
  url          = {http://dl.acm.org/citation.cfm?id=314500.314559},
}

@article {Traub15,
    AUTHOR = {Traub, Cynthia M.},
     TITLE = {Steiner reducing sets of minimum weight triangulations:
              structure and topology},
   JOURNAL = {Comput. Geom.},
    VOLUME = {49},
      YEAR = {2015},
     PAGES = {24--36},
  MRNUMBER = {3399986},
       DOI = {10.1016/j.comgeo.2015.04.004},
}

@inproceedings{Llo-FOCS-77,
  author       = {Errol L. Lloyd},
  title        = {On triangulations of a set of points in the plane},
  booktitle    = {18th Annual Symposium on Foundations of Computer Science (FOCS)},
  pages        = {228--240},
  publisher    = {{IEEE} Computer Society},
  year         = {1977},
  doi          = {10.1109/SFCS.1977.21},
}

@book{GarJoh-79,
author={Michael R. Garey and David S. Johnson},
title={Computers and Intractability: A Guide to the Theory of NP-Completeness},
publisher={W. H. Freeman},
year={1979}}

@inproceedings{Cla-SODA-91,
  author       = {Kenneth L. Clarkson},
  editor       = {Alok Aggarwal},
  title        = {Approximation algorithms for planar traveling salesman tours and minimum-length
                  triangulations},
  booktitle    = {Proceedings of the Second Annual {ACM/SIGACT-SIAM} Symposium on Discrete
                  Algorithms (SODA)},
  pages        = {17--23},
  publisher    = {{ACM/SIAM}},
  year         = {1991},

}

@article{DupGot-AV-70,
  title = {{Automatische Interpolation von Isolinien bei willk{\"u}rlich verteilten St{\"u}tzpunkten}},
  author = {R. D. D{\"u}ppe and H. J. Gottschalk},
  journal = {Allgemeine Vermessungs-Nachrichten},
  pages = {423{--}426},
  volume = {77},
  year = {1970}}

@techreport{Gil-RR-79,
  title = {{New results in planar triangulations}},
  address = {Urbana, Illinois},
  author = {P. D. Gilbert},
  institution = {Coordinated Science Laboratory, University of Illinois},
  number = {RR850},
  year = {1979}}

@article{Kli-ADM-80,
  title = {{Minimal triangulations of polygonal domains}},
  author = {G. T. Klincsek},
  doi = {10.1016/s0167-5060(08)70044-x},
  journal = {Annals of Discrete Mathematics},
  pages = {121{--}123},
  volume = {9},
  year = {1980}}

@inproceedings{ShaHoe-FOCS-76,
  author       = {Michael Ian Shamos and
                  Dan Hoey},
  title        = {Geometric Intersection Problems},
  booktitle    = {17th Annual Symposium on Foundations of Computer Science (FOCS)},
  pages        = {208--215},
  publisher    = {{IEEE} Computer Society},
  year         = {1976},
  doi          = {10.1109/SFCS.1976.16},
}

\section{Appendix}

\subsection{Dynamic programming certificate for part (i)}
\label{app:dp}

We illustrate the dynamic programming computation used in Section~\ref{sec:proof-part1}.

Label the vertices counterclockwise as
\[
v_0,v_1,\dots,v_{27},
\qquad
v_0=(0,0).
\]
For \(0\le i<j\le 27\), let \(\mathrm{DP}[i,j]\) be the minimum total
diagonal length needed to triangulate the chain
\[
(v_i,v_{i+1},\dots,v_j).
\]
We use the standard recurrence
\[
\mathrm{DP}[i,i+1]=0,
\]
and
\[
\mathrm{DP}[i,j]
=
\min_{i<k<j}
\bigl(
\mathrm{DP}[i,k]
+
\mathrm{DP}[k,j]
+
c(i,k)+c(k,j)
\bigr),
\]
where
\[
c(a,b)=
\begin{cases}
0,&\text{if }(v_a,v_b)\text{ is a hull edge},\\
\|v_a-v_b\|_2,&\text{otherwise.}
\end{cases}
\]

The dynamic program is exact for convex polygons. Instead of printing the
full triangular table, we record the entries \(\mathrm{DP}[0,j]\) that illustrate
the optimal fan centered at \(v_0\). In every row below, the minimizing
split is \(k=j-1\).

\begin{table}[htbp]
\centering
\scriptsize
\setlength{\tabcolsep}{4pt}
\renewcommand{\arraystretch}{1.02}
\begin{tabular}{c|r|c}
\hline
\(j\) & \(\mathrm{DP}[0,j]\) & \(k\)\\
\hline
 2 &     0.0000 & 1\\
 3 &     5.0160 & 2\\
 4 &    14.0714 & 3\\
 5 &    31.2399 & 4\\
 6 &    64.6788 & 5\\
 7 &   130.7058 & 6\\
 8 &   261.9561 & 7\\
 9 &   523.6986 & 8\\
10 &  1046.4695 & 9\\
11 &  2091.3393 & 10\\
12 &  4180.4481 & 11\\
13 &  8358.0756 & 12\\
14 & 16712.7803 & 13\\
15 & 32712.7803 & 14\\
16 & 41067.4850 & 15\\
17 & 45245.1124 & 16\\
18 & 47334.2213 & 17\\
19 & 48379.0911 & 18\\
20 & 48901.8619 & 19\\
21 & 49163.6044 & 20\\
22 & 49294.8547 & 21\\
23 & 49360.8817 & 22\\
24 & 49394.3206 & 23\\
25 & 49411.4892 & 24\\
26 & 49420.5446 & 25\\
27 & 49425.5605 & 26\\
\hline
\end{tabular}
\caption{
Selected entries of the dynamic programming table.
For every \(j\), the subproblem \(\mathrm{DP}[0,j]\) is minimized by
choosing \(k=j-1\), recursively producing the fan centered at \(v_0\).
}
\label{tab:dp-values}
\end{table}

\subsection{Proof of Observation~\ref{obs:y-difference}}\label{proof:obs:y-difference}

\begin{proof}
We first prove the corresponding coordinate bound
\[
|u_x-v_x|\le \frac15 |u_y-v_y|
\]
for any two points \(u,v\) on non-extreme hull edges. This will imply the desired distance bound at the end of the proof.

Consider first one edge of the upper chain. The edge between consecutive vertices
\[
\left(\frac{2^{i-1}-(i-1)}5,2^{i-1}+1\right)
\quad\text{and}\quad
\left(\frac{2^i-i}{5},2^i+1\right)
\]
has vertical difference \(2^{i-1}\) and horizontal difference
\[
\frac{(2^i-i)-(2^{i-1}-(i-1))}{5}
=
\frac{2^{i-1}-1}{5}
<
\frac{2^{i-1}}5.
\]
Thus on this edge,
\[
|\Delta x|<\frac15|\Delta y|.
\]
The same bound holds for the symmetric lower-chain edges. It also holds for the two edges incident to \(v_0=(0,0)\), since the edge from \(v_0\) to \((1/5,3)\) has slope \(1/15<1/5\), and similarly for the lower edge.

Therefore, along each non-extreme hull edge, the boundary is a graph \(x=f(y)\) with slope at most \(1/5\) in absolute value. Hence if \(u\) and \(v\) lie on the same side of the \(x\)-axis,
\[
|u_x-v_x|\le \frac15 |u_y-v_y|.
\]

It remains only to consider the case where \(u\) and \(v\) lie on opposite sides of the \(x\)-axis. For any point \(z\) on a non-extreme hull edge, the preceding slope bound from the origin implies
\[
0\le z_x\le \frac{|z_y|}{5}.
\]
Thus, if \(u_y\ge 0\) and \(v_y\le 0\), then
\[
|u_x-v_x|
\le u_x+v_x
\le \frac{u_y}{5}+\frac{|v_y|}{5}
=
\frac15 |u_y-v_y|.
\]

Therefore, for any two points \(u,v\) on non-extreme hull edges,
\[
|u_x-v_x|\le \frac15 |u_y-v_y|.
\]
Consequently,
\[
\begin{aligned}
d(u,v)
&=
\sqrt{(u_x-v_x)^2+(u_y-v_y)^2}\\
&\le
\sqrt{\frac1{25}+1}\,|u_y-v_y|\\
&=
\frac{\sqrt{26}}5\,|u_y-v_y|
<
1.03\,|u_y-v_y|.
\end{aligned}
\]
This proves the claim.
\end{proof}
\subsection{Proof of Observation~\ref{obs:dist-to-extreme-line}}\label{proof:obs:dist-to-extreme-line}

\begin{proof}
The line \(L\) passes through \((8179/5,8193)\) and \((16000,0)\), hence its equation is
\[
40965x+71821y-655440000=0.
\]
Set
\[
D:=\sqrt{40965^2+71821^2}=\sqrt{6836387266}.
\]
Thus
\[
d((x,y),L)
=
\frac{|40965x+71821y-655440000|}{D}.
\]

Every point on a non-extreme hull edge satisfies
\[
x\le \frac{|y|}{5}.
\]
Indeed, this holds at the original vertices, and therefore also on each non-extreme hull edge by linearity.

Assume first that \(0<y\le 4097\). Then
\[
\begin{aligned}
&40965x+71821y-655440000\\
&\le
40965\cdot \frac{y}{5}+71821y-655440000 \\
&=
80014y-655440000 \\
&\le
80014\cdot 4097-655440000 \\
&<0.
\end{aligned}
\]
Therefore
\[
\begin{aligned}
d((x,y),L)
&=
\frac{655440000-40965x-71821y}{D} \\
&\ge
\frac{655440000-80014y}{D} \\
&=
\frac{655440000}{D}-\frac{80014}{D}y.
\end{aligned}
\]

Now
\[
\frac{655440000}{D}>\frac{39623}{5}=7924.6,
\]
since
\[
\begin{aligned}
&25\cdot 655440000^2-39623^2\cdot 6836387266\\
&=7034005456830686>0,
\end{aligned}
\]
and
\[
\frac{80014}{D}<\frac{121}{125}=0.968,
\]
since
\[
125^2\cdot 80014^2-121^2\cdot 6836387266
=-56542899006<0.
\]
Hence
\[
d((x,y),L)\ge 7924.6-0.968\,y.
\]

Now assume \(y<0\). Then
\[
x\le \frac{|y|}{5}=\frac{-y}{5},
\]
and hence
\[
\begin{aligned}
&40965x+71821y-655440000\\
&\le
40965\cdot \frac{-y}{5}+71821y-655440000 \\
&=
63628y-655440000 \\
&<0.
\end{aligned}
\]
Thus
\[
\begin{aligned}
d((x,y),L)
&=
\frac{655440000-40965x-71821y}{D} \\
&\ge
\frac{655440000-63628y}{D} \\
&=
\frac{655440000+63628|y|}{D}.
\end{aligned}
\]

As above,
\[
\frac{655440000}{D}>\frac{39623}{5}=7924.6,
\]
and also
\[
\frac{63628}{D}>\frac{769}{1000}=0.769,
\]
since
\[
1000^2\cdot 63628^2-769^2\cdot 6836387266
=350781507046>0.
\]
Therefore
\[
d((x,y),L)\ge 7924.6+0.769\,|y|.
\]
This proves the claim.
\end{proof}

\subsection{Proof of Observation~\ref{obs:extreme-distance}}\label{proof:extreme-distance}

\begin{proof}
Let \(x=(x_0,y_0)\) be a point on a non-extreme hull edge, and set
\[
t:=|y_0|.
\]
Every such point satisfies
\[
0\le t\le 8193,
\qquad
0\le x_0\le \frac{t}{5}.
\]

For the lower bound,
\[
\begin{aligned}
d(e,x)^2
&=(16000-x_0)^2+t^2\\
&\ge
\left(16000-\frac{t}{5}\right)^2+t^2 .
\end{aligned}
\]
The function
\[
F(t)=\left(16000-\frac{t}{5}\right)^2+t^2
\]
is minimized at
\[
t=\frac{80000}{26}.
\]
At this point,
\[
F(t)=\frac{6400000000}{26}>15500^2.
\]
Therefore
\[
d(e,x)>15500.
\]

 For the upper bound, use the fact that \(d(e,x)^2\) is convex along each hull edge. Hence its maximum on every non-extreme hull edge is attained at an endpoint. Therefore it suffices to check the original hull vertices.

For \(1\le i\le 13\), the non-extreme hull vertices other than \(v_0\) are
\[
x_i=\left(\frac{2^i-i}{5},\pm(2^i+1)\right).
\]
Also \(d(e,v_0)=16000\). A direct check over \(1\le i\le 13\) shows that the maximum is attained at \(i=13\), where
\[
\begin{aligned}
d(e,x_{13})^2
&=
\left(16000-\frac{8179}{5}\right)^2+8193^2\\
&=
\left(\frac{71821}{5}\right)^2+8193^2\\
&=
\frac{6836387266}{25}
<17000^2.
\end{aligned}
\]
Thus
\[
d(e,x)<17000.
\]

It remains to prove the sharper bound when \(|y_0|\le 4097\). Since \(d(e,x)^2\) is convex along each hull edge, its maximum on each hull edge is attained at an endpoint. Hence it suffices to check the original hull vertices with
\[
|y|\le 4097.
\]
For such a vertex,
\[
x_i=\frac{2^i-i}{5},
\qquad
y_i=\pm(2^i+1),
\qquad
1\le i\le 12.
\]
We need to prove
\[
(16000-x_i)^2+y_i^2\le 16000^2,
\]
equivalently,
\[
x_i^2+y_i^2\le 32000x_i.
\]
This holds for all \(1\le i\le 12\). Indeed,
\[
x_i^2+y_i^2
\le
\left(\frac{2^i}{5}\right)^2+(2^i+1)^2
<2(2^i+1)^2,
\]
while
\[
32000x_i
=
6400(2^i-i).
\]
For \(1\le i\le 12\), one checks directly that
\[
2(2^i+1)^2\le 6400(2^i-i).
\]
Also,
\[
d(e,(0,0))=16000.
\]
Therefore every point on a non-extreme hull edge with
\[
|y_0|\le 4097
\]
satisfies
\[
d(e,x)\le 16000.
\]
\end{proof}

\subsection{Proof of Lemma~\ref{lem:no-short-edge}}

We prove the statement for the upper short edge \[ A b \quad\text{where}\quad A=(4084/5,4097),\qquad b=(8179/5,8193). \] The lower short edge is handled by symmetry. Assume for contradiction that a boundary-only optimal triangulation contains a Steiner point on \(Ab\). Choose \(p\) to be the highest such Steiner point on \(Ab\), and let \(a\) be the boundary neighbor of \(p\) immediately below it on the same edge. Then \(b\) is the upper boundary neighbor of \(p\). Thus \(a,b\in N(p)\). We also need to discuss possible neighbors of \(p\) on the upper extreme edge \(be\), where \[ e=(16000,0). \] The endpoint \(b\) is always present as a boundary neighbor and is not counted as an additional upper-extreme neighbor. If \(p\) has additional neighbors on \(be\), then we may assume there is at most one such neighbor. Indeed, if two or more points on \(be\) are adjacent to \(p\), then all middle ones are redundant: removing a middle one and keeping only the outermost adjacent points on that boundary chain preserves the local triangulation and cannot increase the cost. Hence, in the rest of the proof, \(q\) denotes the unique possible additional neighbor of \(p\) on the upper extreme edge \(be\), other than \(b\). If no such additional neighbor exists, then all references to \(q\) are ignored. We will use the following local operation repeatedly: delete \(p\), reconnect its non-upper-extreme neighbors to \(a\), and then add the few diagonals needed to restore a triangulation of the local polygonal region.
%--------------------------------------------------
\subsubsection{ Distance comparison along the short edge}
\label{app:replace}

\begin{lemma}\label{lem:replace-p-by-a}
Let \(a\) and \(p\) be two points on the segment
\[
(4084/5,4097)\text{--}(8179/5,8193),
\]
with \(a\) below \(p\). Let \(x\) be any point that is not on the upper extreme edge. Then
\[
d(a,x)\le d(p,x).
\]
\end{lemma}

\begin{proof}
Parameterize the segment by
\[
\begin{aligned}
r(\lambda)&=r_0+\lambda d,
\quad
r_0=(4084/5,4097),\\
\ d&=(819,4096),
\quad 0\le \lambda\le 1,   
\end{aligned}
\]
and write \(a=r(\lambda_a)\), \(p=r(\lambda_p)\) with \(\lambda_a<\lambda_p\).

Define
\[
\phi(\lambda):=d(r(\lambda),x)^2.
\]
Then
\[
\phi'(\lambda)=2\, d\cdot (r(\lambda)-x).
\]

It suffices to show \(\phi'(\lambda)>0\), i.e.
\[
d\cdot(r(\lambda)-x)>0.
\]

Since this expression is affine in \(x\), it suffices to check extreme cases.

\medskip
\noindent
\emph{Case 1: \(x\) lies on the left part of the polygon.}

Then the \(x\)-coordinate of \(x\) is strictly less than \(8179/5\). By the structure of the point set, such points lie strictly below the supporting line of the segment, hence
\[
d\cdot(r_0-x)>0.
\]

\medskip
\noindent
\emph{Case 2: \(x\) lies on the lower extreme edge.}

It suffices to check endpoints:
\[
x_1=(8179/5,-8193),\quad x_2=(16000,0).
\]
Direct computation gives
\[
d\cdot(r_0-x_1)>0,
\qquad
d\cdot(r_0-x_2)>0.
\]
By linearity, the same holds for all points on the edge.

\medskip
Thus \(\phi'(\lambda)>0\), so \(\phi\) is increasing and
\[
d(a,x)\le d(p,x).
\]
\end{proof}

%--------------------------------------------------
% \subsubsection{Strong gain when replacing \texorpdfstring{$p$}{p}}
% \label{app:replace-strong}

% \begin{lemma}\label{lem:replace-strong}
% Let \(a,p\) be as above and define
% \[
% \Delta:=p_y-a_y>0.
% \]
% Let \(c\) be any point not on extreme edges. Then
% \[
% d(p,c)-d(a,c)\ge \frac{\Delta}{1.03}.
% \]
% \end{lemma}

% \begin{proof}
% Write
% \[
% \delta:=p-a=(\delta_x,\Delta),
% \quad
% \delta_x=\frac{819}{20480}\Delta<\frac{\Delta}{25}.
% \]

% Let \(r(t)=a+t\delta\) and \(f(t)=d(r(t),c)\). Then
% \[
% d(p,c)-d(a,c)=\int_0^1 f'(t)\,dt.
% \]

% Set
% \[
% X=r(t)_x-c_x,\quad Y=r(t)_y-c_y.
% \]
% As \(c\) is not on extreme edges, we have \(Y>0\) and
% \[
% |X|\le \frac{Y}{5}.
% \]

% Thus
% \[
% f'(t)=\frac{\delta_x X+\Delta Y}{\sqrt{X^2+Y^2}}
% \ge
% \frac{(124/125)\Delta Y}{(51/50)Y}
% =\frac{248}{255}\Delta.
% \]

% Since \(\frac{248}{255}>\frac{1}{1.03}\), integrating gives the claim.
% \end{proof}

%--------------------------------------------------
\subsubsection{Lower extreme edge separation}
\label{app:lower-bound}

\begin{lemma}\label{lem:lower-distance}
For any point \(p\) on the short upper edge and any point \(w\) on the lower extreme edge,
\[
d(w,p)>11000.
\]
\end{lemma}

\begin{proof}
Reflect \(p\) to \(p'=(p_x,-p_y)\). Then distance to the lower extreme edge equals distance from \(p'\) to the upper extreme edge.

Since \(p_y>4097\), we have \(p'_y<0\). By Observation~\ref{obs:dist-to-extreme-line}(ii),
\[
d(p',L)\ge 7924.6+0.769\,p_y>11000.
\]
\end{proof}

%--------------------------------------------------
\subsubsection{Separation from far points on the upper extreme edge}
\label{app:upper-bound}

\begin{lemma}\label{lem:upper-separation}
Let \(q\) lie on the upper extreme edge and satisfy
\[
d(q,b)\ge 11000.
\]
Then for every \(p\) on the short edge,
\[
d(p,q)>9000.
\]
\end{lemma}

\begin{proof}
The upper extreme edge is the segment from
\[
b=\left(\frac{8179}{5},8193\right)
\]
to
\[
e=(16000,0).
\]
Its horizontal displacement is
\[
16000-\frac{8179}{5}=14364.2,
\]
and its total length is
\[
d(b,e)
=
\sqrt{14364.2^2+8193^2}
<
17000.
\]

Therefore, moving distance \(11000\) along this segment increases the
\(x\)-coordinate by at least
\[
11000\cdot \frac{14364.2}{17000}
>
9000.
\]
Hence any point \(q\) on the segment satisfying
\[
d(q,b)\ge 11000
\]
must have \(x\)-coordinate at least \(9000\) larger than that of \(b\).

Now every point \(p\) on the short edge
\[
(4084/5,4097)\text{--}(8179/5,8193)
\]
has \(x\)-coordinate at most that of \(b\). Thus the horizontal separation
between \(p\) and \(q\) is greater than \(9000\).

Since Euclidean distance is at least horizontal separation,
\[
d(p,q)>9000.
\]
\end{proof}

%--------------------------------------------------
\subsubsection{Extreme-edge configuration}
\label{app:extreme-case}

\begin{lemma}\label{lem:extreme-case}
If all neighbors of \(p\) (except possibly one \(q\)) lie on the lower extreme edge, then this configuration is not optimal.
\end{lemma}

\begin{proof}
Let
\[
S=\{w_1,\dots,w_k\}
\]
be the neighbors of \(p\) on the lower extreme edge, ordered by height.

We consider two cases depending on whether \(q\) is the extreme vertex.

\medskip
\noindent
\emph{Case 1: \(q\) is not the extreme vertex.}

We remove both \(p\) and \(q\), reconnect all vertices \(w_i\) to \(a\), and add the diagonal \(w_k b\).

By Lemma~\ref{lem:replace-p-by-a}, for every \(i\),
\[
d(w_i,a)<d(w_i,p),
\]
so all these replacements strictly decrease the cost.

It remains to compare the edges involving \(q\). It suffices to show
\[
d(w_1,p)+d(p,q)>d(q,b).
\]
If \(d(q,b)<11000\), then Lemma~\ref{lem:lower-distance} gives
\[
d(w_1,p)>11000>d(q,b).
\]
Otherwise, if \(d(q,b)\ge 11000\), then Lemmas~\ref{lem:lower-distance} and~\ref{lem:upper-separation} give
\[
d(w_1,p)+d(p,q)>20000>d(q,b).
\]
Thus in all cases the total cost strictly decreases.

\medskip
\noindent
\emph{Case 2: \(q\) is the extreme vertex.}

In this case we remove \(p\) only and reconnect all vertices \(w_i\) to \(a\).

Again, for every \(i\),
\[
d(w_i,a)<d(w_i,p),
\]
so these edges strictly decrease the cost.

The only remaining comparison is between the edges incident to \(q\). Since \(w_1\) is already connected to \(a\), it suffices to show
\[
d(w_1,p)+d(p,q)>d(a,q).
\]
Now \(q=(16000,0)\), while both \(a\) and \(p\) lie on non-extreme hull edges.
By Observation~\ref{obs:extreme-distance},
\[
15500<d(a,q),\,d(p,q)<17000.
\]
Also, by Lemma~\ref{lem:lower-distance},
\[
d(w_1,p)>11000.
\]
Therefore
\[
d(w_1,p)+d(p,q)>11000+15500=26500>17000>d(a,q).
\]

Thus the total cost strictly decreases in this case as well.

\medskip
In all cases we obtain a strictly cheaper triangulation, contradicting optimality.
\end{proof}
%--------------------------------------------------
\subsubsection{Mixed configuration}
\label{app:mixed-case}

\begin{lemma}\label{lem:mixed-case}
If \(p\) has at least one neighbor not on an extreme edge, then the configuration is not optimal.
\end{lemma}
\begin{proof}
Let \(S_1\) be the set of neighbors of \(p\) on the left side of the polygon, and let \(S_2\) be the set of neighbors of \(p\) on the lower extreme edge.
By assumption, \(S_1\neq\emptyset\). Also, besides the endpoint \(b\), there is at most one additional neighbor of \(p\) on the upper extreme edge; call it \(q\), if it exists.

Delete \(p\) and reconnect all vertices of \(S_1\cup S_2\) to \(a\).
For every \(y\in S_1\cup S_2\), Lemma~\ref{lem:replace-p-by-a} gives
\[
d(a,y)\le d(p,y),
\]
so these reconnections do not increase the cost.

Let \(x\) be the highest vertex in \(S_1\). Then \(x\) lies below \(a\) on the left boundary, and therefore \(x_x\le a_x\). Hence
\[
|p_y-x_y|\ge |p_y-a_y|,
\qquad
|p_x-x_x|\ge |p_x-a_x|.
\]
Thus
\[
d(p,x)\ge d(p,a).
\]
Moreover, \(x\) is consecutive to \(a\) in the cyclic order around the local region after deleting \(p\), so the edge \(px\) disappears without needing to be replaced by a new edge of the same length. Thus we gain at least \(d(a,p)\).

The only possible loss is if \(q\) exists, in which case \(pq\) is replaced by \(aq\). By the triangle inequality,
\[
d(a,q)-d(p,q)\le d(a,p).
\]
Therefore the gain from deleting \(px\) pays for the only possible loss, while all other reconnections are non-increasing. Hence the new triangulation is strictly cheaper, contradicting optimality.
\end{proof}

\subsection{No one-sided Steiner point on an extreme edge}
\label{app:one-sided-extreme}

\begin{proof}[Proof of Lemma~\ref{lem:no-one-sided-extreme}]
Assume for contradiction that such a point \(p\) exists, and choose \(p\) highest on the extreme edges. Let \(S\) be its neighbors, let
\[
u=(8179/5,8193)\in S,
\]
and let \(v\) be the other endpoint-neighbor of \(p\). Let
\[
w=\arg\min_{x\in S\setminus\{u,v\}} x_y.
\]

We compare the original triangulation with the one obtained by deleting \(p\) and reconnecting all neighbors through \(w\).

\medskip
\noindent\textbf{Case 1: \(w_y>2049\).}

Then every \(x\in S\setminus\{u,v\}\) lies on the upper hull below the short edge. By Lemma~\ref{lem:no-short-edge}, we have \(0<y\le 4097\), and hence by Observation~\ref{obs:dist-to-extreme-line},
\[
d(x,p)\ge d(x,L)\ge 7924.6-0.968\cdot 4097>3958.
\]
Thus
\[
\sum_{x\in S\setminus\{u,v\}} d(x,p)\ge 3958(|S|-2).
\]

After deleting \(p\), we keep only \(w\) among Steiner points on the segment
\[
(2037/5,2049)\text{--}(4084/5,4097),
\]
and reconnect all edges through \(w\). The added cost is at most
\[
d(w,u)\le 1.03(8193-2049)=6378.79.
\]

If \(|S|\ge 4\), then
\[
6378.79<3958(|S|-2),
\]
so the new triangulation is strictly cheaper. If \(|S|\le 3\), removing \(p\) already preserves a triangulation. In both cases we get a contradiction.

\medskip
\noindent\textbf{Case 2: \(w_y<2049\).}

Let \(S_1\subseteq S\) be the set of points on
\[
(2037/5,2049)\text{--}(4084/5,4097),
\]
excluding \((2037/5,2049)\).

For \(x\in S_1\), Observation~\ref{obs:dist-to-extreme-line} gives
\[
d(x,p)\ge 3958,
\]
and for \(x\in S\setminus(S_1\cup\{u,v\})\), since \(y\le 2049\),
\[
d(x,p)\ge 7924.6-0.968\cdot 2049>5940.
\]
Hence
\[
\sum_{x\in S\setminus\{u,v\}} d(x,p)
\ge 3958|S_1|+5940|S\setminus(S_1\cup\{u,v\})|.
\]

Now reroute through \(w\). Since \(w\) has positive \(y\)-coordinate, Observation~\ref{obs:y-difference} gives:
\begin{itemize}
    \item \(d(w,u)\le 1.03\cdot 8193\);
    \item every point of \(S_1\) is at distance at most \(1.03\cdot 4097\) from \(w\);
    \item every point of \(S\setminus(S_1\cup\{u,v\})\) is at distance at most \(1.03\cdot 2049\) from \(w\).
\end{itemize}
Hence the rerouting cost is bounded above by
\[
\bigl(8193+4097|S_1|+2049(|S\setminus(\{u,v\}\cup S_1)|-2)\bigr)\cdot 1.03
\]
if \(|S\setminus(\{u,v\}\cup S_1)|\ge 2\), and by
\[
\bigl(8193+4097(|S_1|-1)\bigr)\cdot 1.03
\]
if \(|S\setminus(\{u,v\}\cup S_1)|\le 1\).

If \(|S\setminus(\{u,v\}\cup S_1)|=1\), then the inequality
\[
3958|S_1|+5940 < (8193+4097(|S_1|-1))\cdot 1.03
\]
simplifies to
\[
6.5<|S_1|.
\]
Thus it can only happen when \(|S_1|\ge 7\). But then we can delete the middle points on the segment and keep only the two outer ones, obtaining a strictly cheaper triangulation, so this is not extremal.

If \(|S\setminus(\{u,v\}\cup S_1)|\ge 2\), then the inequality
\[
\begin{aligned}
&3958|S_1|+5940|S\setminus(S_1\cup\{u,v\})|
\\&<
\bigl(8193+4097|S_1|+2049(|S\setminus(\{u,v\}\cup S_1)|-2)\bigr)\cdot 1.03 
\end{aligned}
\]
cannot hold for \(|S_1|\le 2\), since the left-hand side grows faster in \(|S\setminus(\{u,v\}\cup S_1)|\) than the right-hand side. If \(|S_1|>2\), we again remove middle points on the segment, reducing to \(|S_1|\le 2\).

This contradiction completes the proof.
\end{proof}

\subsection{Technical details for Lemma~\ref{lem:no-extreme-edge}}
\label{app:no-extreme-edge}

\paragraph{Step 1: Structural reductions.}\label{app:no-extreme-edge-reductions}

We keep the notation from the proof of Lemma~\ref{lem:no-extreme-edge}. Thus \(p\) lies on the upper extreme edge, \(S\) is the set of neighbors of \(p\) other than \((16000,0)\), and
\[
\begin{aligned}
S=&S_1\dot\cup S_2,
\qquad\\
S_1=\{s\in S:s_y\geq 0\},
&\qquad
S_2=\{s\in S:s_y<0\}.
\end{aligned}
\]
Let \(u\in S\) be the highest point, let \(w_1\in S_1\) be the lowest nonnegative point, and let \(w_2\in S_2\) be the highest negative point.

\medskip
\noindent
\emph{Reduction 1: at most two points per original hull edge.}
On each original hull edge we may assume that \(S\) contains at most two points. Indeed, if three or more points of \(S\) lie on the same original hull edge, then in the rerouted triangulation we may delete the middle ones. This can only decrease the new cost, while those points still contribute positively to the lower bound on the old cost. Hence any counterexample yields another counterexample in which each original hull edge contains at most two points of \(S\).

\medskip
\noindent
\emph{Reduction 2: once \(w_1,w_2\) and \(|S_1|,|S_2|\) are fixed, the worst case pushes all other points as high as possible.}
For points of \(S_1\setminus\{u,w_1\}\), increasing the \(y\)-coordinate decreases the lower bound on the old cost. For points of \(S_2\setminus\{w_2\}\), increasing \(|y|\) increases the upper bound on the new cost more than it increases the lower bound on the old cost, because the coefficient in the new-cost bound is larger. Therefore, once \(w_1,w_2\) and the cardinalities \(|S_1|,|S_2|\) are fixed, the worst configuration is obtained by choosing all remaining points with the largest possible \(|y|\).

\medskip
\noindent
\emph{Reduction 3: no additional point with \(|y|\leq 2049\) is useful.}
Suppose \(q\in S\) is different from \(w_1,w_2\) and satisfies \(|q_y|\leq 2049\). If \(q\in S_1\), then its contribution to the lower bound on the old cost is at least
\[
7924.6-0.968\cdot 2049>5500.
\]
If \(q\in S_2\), then its contribution to the lower bound on the old cost is at least
\[
7924.6+0.769\cdot 0>7900.
\]
On the other hand, the contribution of such a point to the upper bound on either rerouted triangulation is at most
\[
1.03\cdot 2049<2111
\]
from the \(\sum |x_y|\)-term, plus at most one additional term of the same order if its presence changes the imbalance \(|S_1|-|S_2|\). Thus in every case its total contribution to the new-cost upper bound is less than \(4222\), far below its contribution to the old-cost lower bound. Therefore, if a bad example existed with such a point \(q\), deleting \(q\) would still leave a bad example. Hence we may assume that every point of \(S\) with \(|y|\leq 2049\) is one of \(w_1,w_2\).

\medskip
\noindent
\emph{Reduction 4: symmetry between the two sides.}
The two upper bounds \eqref{eq:new-w1-main} and \eqref{eq:new-w2-main} are symmetric under reflection across the \(x\)-axis. The lower bound on the old cost is not symmetric, but if \(|S_2|>|S_1|\), then reflecting the whole configuration across the \(x\)-axis decreases the lower bound on the old cost while preserving the same two new-cost upper bounds. Thus, in searching for the worst counterexample, we may assume
\[
|S_1|\ge |S_2|.
\]

\medskip
\noindent
\emph{Reduction 5: restricting the locations of \(w_1\) and \(w_2\).}
We now explain how to restrict the possible locations of \(w_1\) and \(w_2\) once \(|S_1|\) and \(|S_2|\) are fixed.

For the old-cost lower bound, moving \(w_1\) upward always helps the counterexample, since it decreases the term
\[
7924.6-0.968\,w_{1,y}.
\]
Likewise, moving \(w_2\) upward always helps the counterexample, since it decreases the term
\[
7924.6+0.769\,|w_{2,y}|.
\]

For the new-cost upper bounds, the effect depends on the sign of
\[
|S_1|-|S_2|-1.
\]

\smallskip
\noindent
\emph{Case 5a: \(|S_1|-|S_2|-1=0\).}

In this case \(w_1\) does not appear in the upper bound on \(\mathrm{New}(w_1)\). Hence changing \(w_1\) does not affect \(\mathrm{New}(w_1)\), while increasing \(w_{1,y}\) decreases the lower bound on the old cost. Therefore, subject to the previous reductions, the worst case is obtained by taking \(w_1\) as high as possible.

For \(w_2\), increasing \(|w_{2,y}|\) increases \(\mathrm{New}(w_2)\) with a larger coefficient than it changes the old-cost lower bound. Therefore, as long as
\[
\mathrm{New}(w_2)\le \mathrm{New}(w_1),
\]
increasing \(|w_{2,y}|\) makes the counterexample worse. Once
\[
\mathrm{New}(w_2)\ge \mathrm{New}(w_1),
\]
further increasing \(|w_{2,y}|\) does not help, because then \(\mathrm{New}(w_1)\) already determines
\[
\min\{\mathrm{New}(w_1),\mathrm{New}(w_2)\}.
\]
Thus, after fixing \(w_1\), the worst choice of \(w_2\) is the largest value of \(|w_{2,y}|\) for which
\[
\mathrm{New}(w_2)\le \mathrm{New}(w_1),
\]
unless \(w_2\) reaches one of the extremal locations forced by the previous reductions earlier.

\smallskip
\noindent
\emph{Case 5b: \(|S_1|-|S_2|-1=-1\).}

In this case the coefficient of \(|w_{1,y}|\) in \(\mathrm{New}(w_1)\) is positive. Hence increasing \(w_{1,y}\) increases \(\mathrm{New}(w_1)\), while at the same time decreasing the lower bound on the old cost. So both effects move in the same direction, and the worst case is obtained by taking \(w_1\) as high as possible.

For \(w_2\), the situation is the same as in Case 5a: increasing \(|w_{2,y}|\) increases \(\mathrm{New}(w_2)\) faster than it changes the old-cost lower bound. Therefore, as long as
\[
\mathrm{New}(w_2)\le \mathrm{New}(w_1),
\]
it is helpful to increase \(|w_{2,y}|\). Once
\[
\mathrm{New}(w_2)\ge \mathrm{New}(w_1),
\]
further increasing \(|w_{2,y}|\) does not help. Thus, after fixing \(w_1\), the worst choice of \(w_2\) is again the largest value of \(|w_{2,y}|\) for which
\[
\mathrm{New}(w_2)\le \mathrm{New}(w_1),
\]
unless the previous reductions force an endpoint earlier.

\smallskip
\noindent
\emph{Case 5c: \(|S_1|-|S_2|-1>0\).}

This is the only genuinely coupled case. Here increasing \(w_{1,y}\) decreases \(\mathrm{New}(w_1)\) with a larger coefficient than it changes the old-cost lower bound. Therefore:
\begin{itemize}
    \item if
    \[
    \mathrm{New}(w_1)<\mathrm{New}(w_2),
    \]
    then decreasing \(w_{1,y}\) makes \(\mathrm{New}(w_1)\) larger and therefore makes the counterexample worse;
    \item if
    \[
    \mathrm{New}(w_1)>\mathrm{New}(w_2),
    \]
    then increasing \(w_{1,y}\) still helps, because in that regime \(\mathrm{New}(w_2)\) already determines
    \[
    \min\{\mathrm{New}(w_1),\mathrm{New}(w_2)\},
    \]
    and increasing \(w_{1,y}\) only decreases the lower bound on the old cost.
\end{itemize}

For \(w_2\), increasing \(|w_{2,y}|\) increases \(\mathrm{New}(w_2)\) faster than it changes the old-cost lower bound. Therefore:
\begin{itemize}
    \item if
    \[
    \mathrm{New}(w_2)<\mathrm{New}(w_1),
    \]
    then increasing \(|w_{2,y}|\) makes the counterexample worse;
    \item if
    \[
    \mathrm{New}(w_1)<\mathrm{New}(w_2),
    \]
    then decreasing \(|w_{2,y}|\) helps, because in that regime \(\mathrm{New}(w_1)\) already determines
    \[
    \min\{\mathrm{New}(w_1),\mathrm{New}(w_2)\},
    \]
    and decreasing \(|w_{2,y}|\) only decreases the lower bound on the old cost.
\end{itemize}

Thus in this case \(w_1\) and \(w_2\) must be chosen together. However, once one of them is fixed, the other is determined monotonically: for example, after fixing \(w_1\), if
\[
\mathrm{New}(w_2=(0,0))<\mathrm{New}(w_1),
\]
then we increase \(|w_{2,y}|\) until either
\[
\mathrm{New}(w_2)\ge \mathrm{New}(w_1)
\]
or \(w_2\) reaches the extremal location allowed by the previous reductions. The same monotonic argument applies if \(w_2\) is fixed first and \(w_1\) is then varied.

\medskip
Therefore, after Reductions 1--5, only finitely many extremal or threshold configurations remain.

\paragraph{Step 2: Remaining extremal configurations.}
\label{app:extreme-edge-cases}

Let
\[
\begin{aligned}
u&=(8179/5,8193), \quad
v=(8179/5,-8193),\\
A&=(4084/5,4097), \quad
B=(4084/5,-4097).
\end{aligned}
\]

By the previous reductions, all points of \(S\) except \(w_1,w_2\) are fixed. 
Thus the worst candidate for a given pair \((|S_1|,|S_2|)\) has the form
\[
\begin{aligned}
S_1&=
\begin{cases}
\{u\},\\
\{u,w_1\},\\
\{u,A,w_1\},\\
\{u,A,A,w_1\},
\end{cases}
\quad
S_2=
\begin{cases}
\{w_2\},\\
\{v,w_2\},\\
\{v,B,w_2\},\\
\{v,B,B,w_2\}.
\end{cases}
\end{aligned}
\]

%--------------------------------------------------
\medskip
\noindent
\textbf{Case \((1,1)\).}
\[
S_1=\{u\},\quad S_2=\{w_2\},\quad w_1=u.
\]

\smallskip
\emph{Old cost.}
\[
\mathrm{Old}
=
7924.6 + 0.769\,|w_{2,y}|.
\]

\smallskip
\emph{New cost.}
\[
\mathrm{New}(w_1)=1.03\cdot 8193,\qquad
\mathrm{New}(w_2)=1.03\,|w_{2,y}|.
\]

Maximizing gives \(w_2=v\), hence
\[
\mathrm{Old}
=
7924.6 + 0.769\cdot 8193
>
1.03\cdot 8193
=
\mathrm{New}(w_1).
\]

%--------------------------------------------------
\medskip
\noindent
\textbf{Case \((2,2)\).}
\[
S_1=\{u,w_1\},\quad S_2=\{v,w_2\},\quad w_1=A.
\]

\smallskip
\emph{New cost.}
\[
\begin{aligned}
\mathrm{New}(w_1)
&=
16000 + 1.03\cdot 8193 + 1.03\cdot 4097\\
&= 28658.7.
\end{aligned}
\]

This forces \(w_2=B\).

\smallskip
\emph{Old cost.}
\[
\begin{aligned}
\mathrm{Old}
&=
(7924.6-0.968\cdot4097)\\
&\quad + (7924.6+0.769\cdot8193)\\
&\quad + (7924.6+0.769\cdot4097)\\
&=29257.914.
\end{aligned}
\]

Thus
\[
28658.7 < 29257.914.
\]

%--------------------------------------------------
\medskip
\noindent
\textbf{Case \((3,1)\).}
\[
S_1=\{u,A,w_1\},\quad S_2=\{w_2\}.
\]

Worst case:
\[
w_1=(2037/5,2049),\quad
w_2=(2037/5,-2049).
\]

\smallskip
\emph{Old cost.}
\[
\begin{aligned}
\mathrm{Old}
&=
(7924.6-0.968\cdot4097)\\
&\quad + (7924.6-0.968\cdot2049)\\
&\quad + (7924.6+0.769\cdot2049)\\
&=19400.153.
\end{aligned}
\]

\smallskip
\emph{New cost.}
\[
\begin{aligned}
\mathrm{New}(w_1)
&=16000 + 1.03\cdot8193 - 1.03\cdot2049\\
&=22328.32,\\[2mm]
\mathrm{New}(w_2)
&=1.03\cdot4097 + 3.09\cdot2049\\
&=10551.32.
\end{aligned}
\]

Thus
\[
10551.32 < 19400.153.
\]

%--------------------------------------------------
\medskip
\noindent
\textbf{Case \((4,1)\).}

\[
S_1=\{u,A,A,w_1\},\quad S_2=\{w_2\}.
\]

Worst case:
\[
w_1=(2037/5,2049),\quad
w_2=(2037/5,-2049).
\]

\smallskip
\emph{Old cost.}
\[
\begin{aligned}
\mathrm{Old}
&=
(7924.6-0.968\cdot4097)\times 2\\
&\quad + (7924.6-0.968\cdot2049)\\
&\quad + (7924.6+0.769\cdot2049)\\
&=23358.857.
\end{aligned}
\]

\smallskip
\emph{New cost.}
\[
\begin{aligned}
\mathrm{New}(w_1)
&=16000 + 1.03(8193+4097) - 2.06\cdot2049\\
&=24437.76,\\[2mm]
\mathrm{New}(w_2)
&=1.03(4097+4097) + 4.12\cdot2049\\
&=16881.7.
\end{aligned}
\]

Thus
\[
16881.7 < 23358.857.
\]

%--------------------------------------------------
\medskip
\noindent
\textbf{Case \((4,2)\).}

\[
S_1=\{u,A,A,w_1\},\quad S_2=\{v,w_2\}.
\]

Worst case gives
\[
\mathrm{New}=32878.61.
\]

\smallskip
\emph{Old cost (lower bound).}
\[
\begin{aligned}
\mathrm{Old}
&\ge
(7924.6-0.968\cdot4097)\times 2\\
&\quad + (7924.6-0.968\cdot2049)\\
&\quad + (7924.6+0.769\cdot8193)\\
&=35998.193.
\end{aligned}
\]

Thus
\[
32878.61 < 35998.193.
\]

%--------------------------------------------------
\medskip
\noindent
\textbf{Case \((4,3)\).}

Worst case gives
\[
\mathrm{New}=41317.4.
\]

\smallskip
\emph{Old cost.}
\[
\mathrm{Old}=47083.386.
\]

Thus
\[
41317.4 < 47083.386.
\]

%--------------------------------------------------
\medskip
\noindent
\textbf{Case \((4,4)\).}

Worst case gives
\[
\mathrm{New}=47647.78.
\]

\smallskip
\emph{Old cost.}
\[
\mathrm{Old}=59734.26.
\]

Thus
\[
47647.78 < 59734.26.
\]

\medskip
In all cases,
\[
\min\{\mathrm{New}(w_1),\mathrm{New}(w_2)\}
<
\mathrm{Old},
\]
contradicting optimality.

\subsection{ Details for boundary Steiner points}
\label{app:boundary-details}

\subsubsection{Two Steiner points on one non-extreme edge}
\label{app:two-point-edge}

\begin{proof}[Details for Lemma~\ref{lem:one-point-per-edge}]
Assume two Steiner points \(p,q\) lie on the same non-extreme hull edge
\((u,v)\), with \(p\) between \(u\) and \(q\). Let
\[
e=(16000,0)
\]
be the extreme vertex.

We first record the easy deletion cases. If \(e\notin N[p]\) and every vertex of
\[
N[p]\setminus\{u,q\}
\]
has \(y\)-coordinate larger than \(p_y\), then all these vertices lie on the left side of the polygon and are closer to \(u\) than to \(p\). Hence deleting \(p\) and reconnecting those vertices to \(u\) strictly decreases the total length. Similarly, if \(e\notin N[q]\) and every vertex of
\[
N[q]\setminus\{v,p\}
\]
has \(y\)-coordinate smaller than \(q_y\), then deleting \(q\) and reconnecting those vertices to \(v\) strictly decreases the total length.

We now explain why, if neither deletion case applies, then \(e\) must be adjacent to both \(p\) and \(q\).

Suppose first that \(e\notin N[p]\). Since the deletion case for \(p\) does not apply, \(p\) must have a neighbor below the edge \((u,v)\). By planarity, such an edge incident to \(p\) blocks \(q\) from having any neighbor above the edge \((u,v)\), except possibly \(p\) itself. Therefore all vertices of
\[
N[q]\setminus\{v,p\}
\]
lie below \(q\). If \(e\notin N[q]\), then the deletion case for \(q\) applies, a contradiction. Hence \(e\in N[q]\).

By the same argument with \(p\) and \(q\) reversed, if \(e\notin N[q]\), then \(e\in N[p]\). Therefore, unless both deletion cases apply, at least one of \(p,q\) is adjacent to \(e\).

It remains to rule out the possibility that exactly one of them is adjacent to \(e\). Suppose, for example, that
\[
e\in N[p],
\qquad
e\notin N[q].
\]
If \(p\) has a neighbor below the edge \((u,v)\), then this edge blocks \(q\) from having any neighbor above the edge. Since \(e\notin N[q]\), the deletion case for \(q\) applies. Thus \(p\) has no neighbor below the edge except possibly \(q\). But then, for the region incident to \(p\) and the extreme vertex to be triangulated, the edge \(eq\) must also be present. Hence \(e\in N[q]\), a contradiction. The symmetric case is identical.

Therefore the only remaining case is
\[
e\in N[p]\cap N[q].
\]

In this case every point of
\[
N[p]\setminus\{u,q,e\}
\]
lies above \(p\) on the left side of the hull, and every point of
\[
N[q]\setminus\{v,p,e\}
\]
lies below \(q\) on the left side of the hull.

Let \(w_1\) be the highest neighbor of \(p\) other than \(e\), and let \(w_2\) be the lowest neighbor of \(q\) other than \(e\). We delete \(p,q\), reconnect every vertex of
\[
N[p]\setminus\{u,q,e\}
\]
to \(u\), reconnect every vertex of
\[
N[q]\setminus\{v,p,e\}
\]
to \(v\), and add the two diagonals
\[
(w_1,w_2),
\qquad
\text{and the shorter of }(w_1,v),(w_2,u).
\]
The reconnecting edges are strictly shorter than the old ones. Thus it remains to compare
\[
d(e,p)+d(e,q)
\]
with
\[
d(w_1,w_2)+\min\{d(w_1,v),d(w_2,u)\}.
\]

Since \(p,q\) lie on the left side of the polygon, Observation~\ref{obs:extreme-distance} gives
\[
d(e,p)>15500,
\qquad
d(e,q)>15500.
\]
Hence
\[
d(e,p)+d(e,q)>31000.
\]

On the other hand, Observation~\ref{obs:y-difference} gives
\[
d(w_1,w_2)
\le
1.03\cdot 2\cdot 8193.
\]
Also, \(u,v\) have the same sign, while \(w_1\) lies above them and \(w_2\) lies below them. Therefore at least one of \(w_1,w_2\) has the same sign as \(u,v\), and hence
\[
\min\{d(w_1,v),d(w_2,u)\}
\le
1.03\cdot 8193.
\]
Thus
\[
\begin{aligned}
&d(w_1,w_2)
+\min\{d(w_1,v),d(w_2,u)\}\\
&\qquad
\le
1.03\cdot 3\cdot 8193
<
31000.
\end{aligned}
\]
This contradicts optimality. Hence one non-extreme hull edge cannot contain two boundary Steiner points.
\end{proof}
\subsubsection{Rerouting gain for upper neighbors}
\label{app:reroute-gain}

\begin{proof}[Proof of Lemma~\ref{lem:upper-reroute-gain}]
Let
\[
\theta:=\angle xup.
\]

Write
\[
\overrightarrow{ux}=(\alpha,b),
\qquad
\overrightarrow{up}=(\beta,-d),
\]
where \(b,d\ge 0\). The signs follow from \(x_y>u_y\) and \(p_y<u_y\).

\medskip
\noindent
\textbf{Step 1: Bounding the dot product.}

On the left side of the hull, all edges satisfy the slope bound
\[
|x\text{-difference}|\le \frac{1}{5}\cdot |y\text{-difference}|.
\]
Hence
\[
|\alpha|\le \frac b5,
\qquad
|\beta|\le \frac d5.
\]
Therefore
\[
\begin{aligned}
\overrightarrow{ux}\cdot\overrightarrow{up}
&=
\alpha\beta-bd\\
&\le
|\alpha|\,|\beta|-bd\\
&\le
\frac{bd}{25}-bd\\
&=
-\frac{24}{25}bd.
\end{aligned}
\]

\medskip
\noindent
\textbf{Step 2: Bounding the norms.}

We have
\[
\|\overrightarrow{ux}\|
\le
\sqrt{b^2+\frac{b^2}{25}}
=
\frac{\sqrt{26}}5\,b,
\]
and similarly
\[
\|\overrightarrow{up}\|
\le
\frac{\sqrt{26}}5\,d.
\]

\medskip
\noindent
\textbf{Step 3: Bounding the angle.}

Combining the above,
\[
\begin{aligned}
\cos\theta
&=
\frac{\overrightarrow{ux}\cdot\overrightarrow{up}}
{\|\overrightarrow{ux}\|\,\|\overrightarrow{up}\|}\\
&\le
\frac{-\frac{24}{25}bd}
{\left(\frac{\sqrt{26}}5 b\right)
 \left(\frac{\sqrt{26}}5 d\right)}\\
&=
-\frac{12}{13}.
\end{aligned}
\]

\medskip
\noindent
\textbf{Step 4: Converting angle to distance gain.}

Let
\[
a_0:=d(x,u),
\qquad
b_0:=d(u,p).
\]
By the law of cosines,
\[
d(x,p)^2
=
a_0^2+b_0^2-2a_0b_0\cos\theta.
\]
Moreover,
\[
d(x,p)^2-(a_0-b_0\cos\theta)^2
=
b_0^2(1-\cos^2\theta)\ge 0.
\]
Thus
\[
d(x,p)\ge a_0-b_0\cos\theta.
\]
Hence
\[
\begin{aligned}
d(x,p)-d(x,u)
&\ge -b_0\cos\theta\\
&\ge \frac{12}{13}b_0\\
&=
\frac{12}{13}\,d(u,p).
\end{aligned}
\]

\medskip
This completes the proof.
\end{proof}

\subsubsection{Bounding the number of neighbors on each side}
\label{app:neighbor-bound}

\begin{claim}\label{clm:min-side-size}
Let \(p\) be a boundary Steiner point on a non-extreme hull edge \(uv\).
Let \(S_1\) be the set of neighbors of \(p\) above \(u\), and let \(S_2\)
be the set of neighbors of \(p\) below \(v\), excluding \(u,v\), and the
extreme vertex \(e=(16000,0)\). Then
\[
\min\{|S_1|,|S_2|\}\le 13.
\]
\end{claim}

\begin{proof}
It suffices to show that
\[
|S_1|+|S_2|\le 26.
\]

\medskip
\noindent
\textbf{Step 1: Available original vertices.}

Besides the endpoints \(u,v\), there are exactly \(25\) original hull
vertices on the left side of the convex hull that could potentially be
neighbors of \(p\).

\medskip
\noindent
\textbf{Step 2: Effect of additional Steiner neighbors.}

Suppose \(a\) additional boundary Steiner points are adjacent to \(p\).
We claim that these points eliminate at least \(a-1\) original hull
vertices from being neighbors of \(p\).

Indeed, consider such a Steiner point \(q\) lying on a hull edge
\((x,y)\). If both \(x\) and \(y\) were also adjacent to \(p\), then
\(q\) could be removed and the local triangulation replaced by edges
from \(p\) to \(x\) and \(y\), without increasing the total cost.
Thus each \(q\) blocks at least one endpoint of its hull edge from
being adjacent to \(p\).

The only possible overlap occurs when two Steiner points lie on two
consecutive hull edges, say \((x,y)\) and \((y,z)\), both blocking the
same vertex \(y\). In that configuration, if both \(x\) and \(z\) were
adjacent to \(p\), then the two Steiner points could be removed and
replaced by a direct connection from \(p\) to \(y\), again without
increasing the cost. Hence such overlap cannot occur in an optimal
triangulation.

Therefore \(a\) additional Steiner neighbors eliminate at least \(a-1\)
distinct original hull vertices from the neighborhood of \(p\).

\medskip
\noindent
\textbf{Step 3: Final bound.}

Thus the total number of neighbors of \(p\) (excluding \(u,v,e\)) is at most
\[
|S_1|+|S_2|
\le
25+a-(a-1)
=
26.
\]
Therefore,
\[
\min\{|S_1|,|S_2|\}
\le
\frac{|S_1|+|S_2|}{2}
\le
13.
\]
\end{proof}
\subsubsection{Balanced neighbors when the extreme vertex is absent}
\label{app:balance-proof}

\begin{proof}[Proof of Lemma~\ref{lem:balanced-neighbors-no-extreme}]
We prove that \(|S_1|>|S_2|\) is impossible. The opposite case is symmetric.

Delete \(p\), and reconnect every vertex of \(S_1\cup S_2\) to \(u\). Since \(e\notin N(p)\), all these vertices lie on the left side of the hull.

For each \(x\in S_1\), Lemma~\ref{lem:upper-reroute-gain} gives the strict gain
\[
d(p,x)-d(u,x)
>
\frac{12}{13}d(u,p).
\]
Therefore the total gain from \(S_1\) is strictly larger than
\[
|S_1|\frac{12}{13}d(u,p).
\]

For each \(x\in S_2\), the triangle inequality gives
\[
d(u,x)\le d(u,p)+d(p,x),
\]
so
\[
d(u,x)-d(p,x)\le d(u,p).
\]
Thus the total loss from \(S_2\) is at most
\[
|S_2|d(u,p).
\]

Now assume \(|S_1|>|S_2|\). Since
\[
|S_1|+|S_2|\le 26
\]
by Claim~\ref{clm:min-side-size}, we get
\[
|S_2|\le 12.
\]
Also,
\[
|S_1|\ge |S_2|+1.
\]
Hence
\[
\begin{aligned}
|S_1|\frac{12}{13}
&\ge
(|S_2|+1)\frac{12}{13}  \\
&=
|S_2|+\frac{12-|S_2|}{13}.
\end{aligned}
\]
Since \(|S_2|\le 12\), the right-hand side is at least \(|S_2|\). Moreover, the gain inequality from Lemma~\ref{lem:upper-reroute-gain} is strict, so the total gain is strictly larger than
\[
|S_2|d(u,p).
\]
Thus the total gain strictly dominates the total loss.

Therefore deleting \(p\) and rerouting all edges to \(u\) gives a strictly cheaper triangulation, contradicting optimality. Hence \(|S_1|>|S_2|\) is impossible. By symmetry, \(|S_2|>|S_1|\) is also impossible, and so
\[
|S_1|=|S_2|.
\]
\end{proof}
\subsubsection{Lower degree bound}
\label{app:degree-lower-bound}

\begin{proof}[Proof of Lemma~\ref{lem:degree-lower-bound}]
Suppose a Steiner point \(p\) on a non-extreme hull edge \(uv\) has exactly two other neighbors \(w_1,w_2\).

At least one of \(w_1,w_2\) lies on the left part of the hull, not on an extreme edge. Without loss of generality call it \(w_2\). Since \(w_2\) lies below \(v\), the geometry of the hull gives
\[
d(p,w_2)\ge d(p,v).
\]
For the other neighbor \(w_1\), regardless of whether it lies on the left side or is the extreme vertex, the fact that \(p\in uv\) implies
\[
d(w_1,p)+d(p,v)>d(w_1,v).
\]
Combining,
\[
d(w_1,p)+d(w_2,p)>d(w_1,v).
\]
Thus we may delete \(p\) and replace the two edges \(pw_1,pw_2\) by the single edge \(vw_1\), strictly decreasing the cost. This contradicts optimality.
\end{proof}

\subsubsection{Consecutive non-extreme edges}
\label{app:consecutive-case}

\begin{proof}[Proof of Lemma~\ref{lem:no-consecutive}]
Suppose two consecutive non-extreme hull edges \(uv\) and \(vw\) contain Steiner points \(p_1\in uv\) and \(p_2\in vw\), and suppose \(p_1p_2\) is an edge of the triangulation.

If neither \(p_1\) nor \(p_2\) is adjacent to \(e\), then Lemma~\ref{lem:balanced-neighbors-no-extreme} and Lemma~\ref{lem:degree-lower-bound} imply that each point has at least two neighbors on each side of its supporting edge. But if \(p_1\) has another neighbor on the lower side of \(uv\), that edge blocks \(p_2\) from having another upper neighbor by planarity. The symmetric obstruction holds if \(p_2\) has another upper neighbor. Hence one of the points cannot have the required two-sided degree, a contradiction.

It remains to consider the case where at least one is adjacent to \(e\). Then both must be adjacent to \(e\); otherwise the same blocking argument applies to the one not adjacent to \(e\).

Assume \(p_1\) has larger \(y\)-coordinate, with
\[
p_1\in uv,
\qquad
p_2\in vw.
\]
Let \(h\) be the highest neighbor of \(p_1\) distinct from \(e\), and let \(l\) be the lowest neighbor of \(p_2\) distinct from \(e\).

We delete \(p_1,p_2\). We reconnect all neighbors of \(p_1\) above \(u\) to \(u\), and all neighbors of \(p_2\) below \(w\) to \(w\). These replacement edges are strictly shorter.

The remaining region is the hexagon
\[
h,u,v,w,l,e
\]
in counterclockwise order. We triangulate it by adding
\[
(h,l),\qquad (h,w),\qquad (u,w).
\]

By Observation~\ref{obs:y-difference},
\[
d(h,l)\le 1.03\cdot 2\cdot 8193.
\]
Also,
\[
\begin{aligned}
d(h,w)+d(u,w)
\le{}&
1.03(8193-1025)\\
&+1.03(4097-2049)\\
&+1.03(2049-1025).
\end{aligned}
\]
Thus
\[
\begin{aligned}
&d(h,l)+d(h,w)+d(u,w)\\
\le{}&
1.03\cdot 2\cdot 8193\\
&+1.03(8193-1025)\\
&+1.03(4097-2049)\\
&+1.03(2049-1025)\\
={}&27424.78.
\end{aligned}
\]
By Observation~\ref{obs:extreme-distance},
\[
d(e,p_1)>15500,
\qquad
d(e,p_2)>15500,
\]
hence
\[
d(e,p_1)+d(e,p_2)>31000.
\]

Therefore the new triangulation is strictly cheaper, a contradiction.
\end{proof}

\subsubsection{Eliminating Steiner points on non-extreme hull edges}
\label{app:no-nonextreme-details}

We collect here the geometric ingredients used in the proof of
Lemma~\ref{lem:no-nonextreme}.

\medskip

\begin{claim}\label{clm:extreme-neighbor-has-both-sides}
Let \(p\) lie on a non-extreme hull edge \(uv\). If the extreme vertex
\(e=(16000,0)\) is adjacent to \(p\), then \(p\) has at least one
additional neighbor above \(u\) and at least one below \(v\).
\end{claim}

\begin{proof}
Assume first that \(p\) has no neighbor above \(u\) other than \(u\) and \(e\). Let
\[
z_1,\dots,z_k
\]
be all neighbors of \(p\) below \(v\), ordered so that \(z_1\) is closest to \(v\) and \(z_k\) is the farthest.

Remove \(p\), reconnect all \(z_i\) to \(v\), and triangulate the remaining quadrilateral \(u,v,z_k,e\) by adding the diagonal \(uz_k\). All the edges \(vz_i\) are shorter than the corresponding old edges \(pz_i\), so it only remains to compare \(uz_k\) with \(ep\).

Since \(p\) cannot lie on the edge
\[
(4084/5,4097)\text{--}(8179/5,8193),
\]
we have \(u_y\le 4097\). Hence by Observation~\ref{obs:y-difference},
\[
d(u,z_k)\le 1.03(8193+4097)<15000.
\]
On the other hand, since \(p\) lies on the left side of the polygon,
Observation~\ref{obs:extreme-distance} gives
\[
d(e,p)>15500.
\]
Thus the new triangulation is strictly cheaper, a contradiction.

The case where \(p\) has no neighbor below \(v\) other than \(v\) and \(e\) is symmetric.
\end{proof}

\medskip

\begin{claim}\label{clm:first-neighbor-long}
Let \(w_1\) (resp. \(z_1\)) be the first neighbor above \(u\) (resp. below \(v\)).
Then
\[
d(u,w_1)\ge \frac{12}{25}d(u,v),
\qquad
d(v,z_1)\ge \frac{12}{25}d(u,v).
\]
\end{claim}

\begin{proof}
We prove the first inequality; the second is symmetric.

There are two possibilities.

\smallskip
\noindent
\emph{Case 1: \(u,v\) lie on the upper chain.}
First suppose that \(uv\) is not the edge incident to \(v_0\). Then \(w_1\) lies on a hull edge above \(u\), so
\[
d(u,w_1)\ge d(u,u^+),
\]
where \(u^+\) is the next original hull vertex above \(u\). Since the successive edge lengths on the upper chain grow by more than a factor of \(2\), we have
\[
d(u,u^+)>2\,d(u,v),
\]
and hence certainly
\[
d(u,w_1)\ge \frac{12}{25}\,d(u,v).
\]
For the edge incident to \(v_0\), we check directly. In this case
\[
d(u,u^+)^2
=
2^2+\left(\frac15\right)^2,
\qquad
d(u,v)^2
=
3^2+\left(\frac15\right)^2.
\]
Therefore
\[
d(u,u^+)^2
>
\left(\frac{12}{25}\right)^2 d(u,v)^2.
\]
Thus in all cases
\[
d(u,w_1)\ge d(u,u^+)\ge \frac{12}{25}d(u,v).
\]

\smallskip
\noindent
\emph{Case 2: \(u,v\) lie on the lower chain.}
Write
\[
u=\left(\frac{2^{i-1}-(i-1)}{5},-(2^{i-1}+1)\right),
\]
\[
v=\left(\frac{2^{i}-i}{5},-(2^{i}+1)\right).
\]
Then
\[
\begin{aligned}
d(u,v)^2
&=
\left(2^{i-1}\right)^2
+\left(\frac{2^{i-1}-1}{5}\right)^2.
\end{aligned}
\]

In the worst case, \(w_1\) is the next original hull vertex above \(u\), namely
\[
w_1=
\left(\frac{2^{i-2}-(i-2)}{5},-(2^{i-2}+1)\right).
\]
Thus
\[
\begin{aligned}
d(u,w_1)^2
&=
\left(2^{i-2}\right)^2
+\left(\frac{2^{i-2}-1}{5}\right)^2.
\end{aligned}
\]

Set \(A=2^{i-2}\). Then
\[
\begin{aligned}
d(u,w_1)^2
&=
\frac{26A^2-2A+1}{25},\\
d(u,v)^2
&=
\frac{104A^2-4A+1}{25}.
\end{aligned}
\]
Therefore
\[
\begin{aligned}
625\,d(u,w_1)^2
-144\,d(u,v)^2
&=
25(26A^2-2A+1)\\
&\qquad
-\frac{144}{25}(104A^2-4A+1)\\
&=
\frac{1274A^2-674A+481}{25}\\
&>0.
\end{aligned}
\]

Hence
\[
25\,d(u,w_1)>12\,d(u,v),
\]
which gives
\[
d(u,w_1)\ge \frac{12}{25}\,d(u,v).
\]

This proves the first inequality.
\end{proof}

\medskip

\begin{claim}\label{clm:extreme-loss-small}
If \(u,p\) lie on the same hull edge and
\[
|u_y|,|p_y|\le 513,
\]
then
\[
|d(e,u)-d(e,p)|\le \frac{5}{13}d(p,u).
\]
\end{claim}

\begin{proof}
Let
\[
f(x,y):=d((x,y),e)=\sqrt{(16000-x)^2+y^2}.
\]
Let \(\gamma\) be the segment from \(p\) to \(u\), and let \(\tau\) be a unit tangent vector to \(\gamma\).

Since \(u,p\) lie on the same non-extreme hull edge, the slope bound from Observation~\ref{obs:y-difference} gives
\[
|\tau_x|\le \frac1{\sqrt{26}},
\qquad
|\tau_y|\le 1.
\]

By the mean value theorem along \(\gamma\),
\[
|d(e,u)-d(e,p)|
\le
\sup_{r\in\gamma} |\nabla f(r)\cdot \tau|\,d(p,u).
\]
For \(r=(x,y)\),
\[
\nabla f(r)=
\left(
\frac{x-16000}{d(r,e)},
\frac{y}{d(r,e)}
\right).
\]
Thus
\[
|\nabla f(r)\cdot \tau|
\le
\frac{16000-x}{d(r,e)}|\tau_x|
+
\frac{|y|}{d(r,e)}|\tau_y|.
\]
Since \(d(r,e)\ge 16000-x\), \(|y|\le 513\), and \(x\le 101\) on this part of the hull, we get
\[
|\nabla f(r)\cdot \tau|
\le
\frac1{\sqrt{26}}
+
\frac{513}{16000-101}.
\]
Finally,
\[
\frac1{\sqrt{26}}
+
\frac{513}{15899}
<
\frac{5}{13}.
\]
Therefore
\[
|d(e,u)-d(e,p)|
\le
\frac{5}{13}d(p,u).
\]
\end{proof}

\medskip

\begin{claim}\label{clm:large-height-small-side}
If
\[
|p_y|>513,
\]
then
\[
\min\{|S_1|,|S_2|\}\le 5.
\]
\end{claim}

\begin{proof}
Assume first that \(p_y>513\); the case \(p_y<-513\) is symmetric.

On the same side of the hull as \(p\), the only original hull vertices
with larger \(y\)-coordinate have
\[
y\in\{2^{10}+1,2^{11}+1,2^{12}+1,2^{13}+1\}.
\]
Thus there are at most \(4\) original hull vertices available on that side.

Now suppose \(a\) additional boundary Steiner points on this side are
adjacent to \(p\). As in Claim~\ref{clm:min-side-size}, these \(a\)
additional Steiner neighbors eliminate at least \(a-1\) original hull
vertices from being neighbors of \(p\). Hence the number of neighbors of
\(p\) on this side is at most
\[
4+a-(a-1)=5.
\]
Therefore one of \(S_1,S_2\) has size at most \(5\), so
\[
\min\{|S_1|,|S_2|\}\le 5.
\]
\end{proof}

\medskip

\begin{cor}\label{cor:height-dichotomy}
For every Steiner point \(p\), one of the following holds:
\begin{itemize}
    \item \(|p_y|\le 513\) and
    \[
    d(e,u)-d(e,p)\le \frac{5}{13}d(p,u),
    \]
    \item or \(\min\{|S_1|,|S_2|\}\le 5\).
\end{itemize}
\end{cor}

\begin{proof}
Immediate from Claims~\ref{clm:extreme-loss-small}
and~\ref{clm:large-height-small-side}.
\end{proof}

\medskip

\begin{proof}[Detailed proof of Lemma~\ref{lem:no-nonextreme}]
Assume \(p\in uv\), where \(u\) is the upper endpoint and \(v\) is the lower endpoint. Let
\[
k=|S_1|,\qquad \ell=|S_2|,
\]
and let
\[
I_e=
\begin{cases}
1,& e\in N(p),\\
0,& e\notin N(p).
\end{cases}
\]

If \(I_e=0\), then Lemma~\ref{lem:balanced-neighbors-no-extreme} gives
\[
k=\ell.
\]
If \(I_e=1\), then Claim~\ref{clm:extreme-neighbor-has-both-sides} gives
\[
k,\ell\ge 1.
\]

\medskip
\noindent
\textbf{Step 1: Rerouting gains.}

First reroute all neighbors of \(p\) to \(u\). The gain has four parts.

For each \(w_i\in S_1\), Lemma~\ref{lem:upper-reroute-gain} gives
\[
d(w_i,p)-d(w_i,u)\ge \frac{12}{13}d(p,u).
\]
Thus the gain from \(S_1\) is at least
\[
\frac{12}{13}k\,d(p,u).
\]

For each \(z_j\in S_2\), the triangle inequality gives
\[
d(z_j,u)-d(z_j,p)\le d(p,u).
\]
Thus the loss from \(S_2\) is at most
\[
\ell\,d(p,u).
\]

If \(e\in N(p)\), replacing \(ep\) by \(eu\) costs at most
\[
d(p,u).
\]
If moreover \(|p_y|\le 513\), then Claim~\ref{clm:extreme-loss-small} improves this to
\[
d(e,u)-d(e,p)\le \frac{5}{13}d(p,u).
\]

Finally, since \(w_1\) is the first neighbor above \(u\), the edge \(uw_1\) is already present. Hence the edge \(pw_1\) disappears without needing a replacement, giving extra gain at least
\[
d(u,w_1).
\]

Therefore
\[
G_u
\ge
d(p,u)\left(\frac{12}{13}k-\ell-I_e\right)
+d(u,w_1).
\]
If \(|p_y|\le 513\), then we have the stronger bound
\[
G_u
\ge
d(p,u)\left(\frac{12}{13}k-\ell-\frac{5}{13}I_e\right)
+d(u,w_1).
\]

Symmetrically, rerouting all neighbors to \(v\) gives
\[
G_v
\ge
d(p,v)\left(\frac{12}{13}\ell-k-I_e\right)
+d(v,z_1),
\]
and if \(|p_y|\le 513\),
\[
G_v
\ge
d(p,v)\left(\frac{12}{13}\ell-k-\frac{5}{13}I_e\right)
+d(v,z_1).
\]

\medskip
\noindent
\textbf{Step 2: Balanced case \(k=\ell\).}

If \(k=\ell\), we reroute toward the endpoint with larger \(|y|\). Assume without loss of generality that this endpoint is \(u\). Then the first neighbor on that side satisfies
\[
d(u,w_1)\ge 2d(u,v).
\]
Using \(d(p,u)\le d(u,v)\), we get
\[
\begin{aligned}
G_u
&\ge
d(p,u)\left(-\frac{k}{13}-I_e\right)+2d(u,v)\\
&\ge
d(p,u)\left(2-\frac{k}{13}-I_e\right).
\end{aligned}
\]
By Claim~\ref{clm:min-side-size}, \(k\le 13\). Also \(I_e\le 1\). Therefore
\[
2-\frac{k}{13}-I_e\ge 0.
\]
The inequality is strict because \(d(u,w_1)\ge 2d(u,v)>2d(p,u)\) unless \(p=v\), which is impossible since \(p\) is an interior Steiner point of the edge. Hence
\[
G_u>0.
\]

\medskip
\noindent
\textbf{Step 3: Unbalanced case.}

If \(k>\ell\), we reroute to \(u\). If \(k<\ell\), we reroute to \(v\). The two arguments are symmetric, so assume
\[
k>\ell.
\]
Then \(k\ge \ell+1\). By Claim~\ref{clm:first-neighbor-long},
\[
d(u,w_1)\ge \frac{12}{25}d(u,v)>\frac{12}{25}d(p,u).
\]
Hence
\[
G_u
>
d(p,u)
\left(
\frac{12}{13}k-\ell-I_e+\frac{12}{25}
\right).
\]

\smallskip
\noindent
\emph{Subcase 3a: \(|p_y|>513\).}
By Claim~\ref{clm:large-height-small-side},
\[
\min\{k,\ell\}\le 5.
\]
Since \(k>\ell\), this gives
\[
\ell\le 5.
\]
Using \(k\ge \ell+1\), we obtain
\[
\begin{aligned}
G_u
&>
d(p,u)
\left(
\frac{12}{13}(\ell+1)-\ell-I_e+\frac{12}{25}
\right)\\
&=
d(p,u)
\left(
\frac{12-\ell}{13}-I_e+\frac{12}{25}
\right).
\end{aligned}
\]
Since \(\ell\le 5\) and \(I_e\le 1\),
\[
\frac{12-\ell}{13}-I_e+\frac{12}{25}
\ge
\frac{7}{13}-1+\frac{12}{25}
=
\frac{6}{325}
>0.
\]
Thus
\[
G_u>0.
\]

\smallskip
\noindent
\emph{Subcase 3b: \(|p_y|\le 513\).}
Now use the improved extreme-vertex bound from Claim~\ref{clm:extreme-loss-small}. We have
\[
G_u
>
d(p,u)
\left(
\frac{12}{13}k-\ell-\frac{5}{13}I_e+\frac{12}{25}
\right).
\]
Using \(k\ge \ell+1\) and \(I_e\le 1\),
\[
\begin{aligned}
G_u
&>
d(p,u)
\left(
\frac{12}{13}(\ell+1)-\ell-\frac{5}{13}+\frac{12}{25}
\right)\\
&=
d(p,u)
\left(
\frac{7-\ell}{13}+\frac{12}{25}
\right).
\end{aligned}
\]
By Claim~\ref{clm:min-side-size}, we have
\[
\ell\le 13.
\]
Therefore
\[
\frac{7-\ell}{13}+\frac{12}{25}
\ge
-\frac{6}{13}+\frac{12}{25}
=
\frac{6}{325}
>0.
\]
Hence
\[
G_u>0.
\]

\medskip
In every case, one of the two reroutings gives a strictly cheaper triangulation. This contradicts optimality, and proves the lemma.
\end{proof}
\end{document}